\documentclass{article}
\usepackage[margin=0.62in,a4paper]{geometry}
\usepackage{moresize}
\usepackage{algorithmicx}
\usepackage{algorithm}
\usepackage{algpseudocode}
\usepackage{float}
\usepackage{graphicx}
\usepackage[
backend=biber,
style=ieee,
sorting=none
]{biblatex}
\usepackage[colorlinks,linkcolor=blue]{hyperref}
\usepackage{tikz}
\usetikzlibrary{arrows.meta}
\tikzset{
	>={Latex[width=2mm,length=2mm]},
	base/.style = {rectangle, rounded corners, draw=black,
		minimum width=4cm, minimum height=1cm,
		text centered, font=\sffamily},
	process/.style = {base, minimum width=1cm, fill=orange!15,
		font=\ttfamily},
}
\usepackage{amsmath, amsthm,amssymb}
\newtheorem{theorem}{Theorem}[section]

\newtheorem{innercustomgeneric}{\customgenericname}
\providecommand{\customgenericname}{}
\newcommand{\newcustomtheorem}[2]{%
  \newenvironment{#1}[1]
  {%
   \renewcommand\customgenericname{#2}%
   \renewcommand\theinnercustomgeneric{##1}%
   \innercustomgeneric
  }
  {\endinnercustomgeneric}
}

\newcustomtheorem{customthm}{Theorem}
\title{Data Forecasts of the Epidemic COVID-19 by Deterministic and Stochastic Time-Dependent Models}
\usepackage[utf8]{inputenc}
\usepackage{authblk}
\usepackage{cleveref}
\author[$*$]{Bo-Sheng Chen}
\author[$**$]{Zong-Ying Wu}
\author[$\dagger$]{Yen-Jia Chen}
\author[$\ddagger$]{Jann-Long Chern}

\affil[$*$]{\small{Corresponding author: Department of Mathematics, National Taiwan Normal University, Taipei, Taiwan.}}
\affil[ ]{ The authors are supported in part by National Science and Technology Council (NSTC), Taiwan,}
\affil[ ]{
No. NSTC 111-2813-C-003-073-M. \tt 40840103s@ntnu.edu.tw}
\affil[$**$]{\small{Department of Mathematics, National Taiwan Normal University, Taipei 11677, Taiwan. \tt 40840317s@ntnu.edu.tw}}
\affil[$\dagger$]{Data Science Degree Program, National Taiwan University and Academia Sinica, Taipei, Taiwan. \tt r10946007@ntu.edu.tw}
\affil[$\ddagger$]{\small{Department of Mathematics, National Taiwan Normal University, Taipei 11677, Taiwan.}}
\affil[ ]{The authors are supported in part by National Science and Technology Council (NSTC), Taiwan,}
\affil[ ]{
No. NSTC 110-2115-M-003-019-MY3 and NSTC 111-2218-E-008-004-MBK. \tt chern@math.ntnu.edu.tw}
\date{}
\begin{document}
\maketitle
\begin{abstract}
We propose a deterministic SAIVRD model and a stochastic SARV model of the epidemic COVID-19 involving asymptomatic infections and vaccinations to conduct data forecasts using time-dependent parameters. The forecast by our deterministic model conducts 10-day predictions to see whether the epidemic will ease or become more severe in the short term. The forecast by our stochastic model predicts the probability distributions of the final size and the maximum size to see how large the epidemic will be in the long run. The first forecast using the data set from the USA gives the relative errors within $3\%$ in $5$ days and $7\%$ in $10$ days for the prediction of isolated infectious cases. The distributions in the second forecast using the time-varying parameters from the first forecast are bimodal in our model with time-independent parameters in our simulations of smaller populations. For the model with time-dependent model, what are different are that there is another peak in the final size distribution, that the probability of minor outbreak is higher and that the maximum size distribution is oscillating with time-dependent parameters. The final size distributions are similar between different populations and so are the maximum size distributions, which means that we can expect that with the same parameters and
in a large population, the ratio of the final size and the maximum size are distributed similarly. The result shows that  under recent transmissibility of this disease in the USA, when an initial infection is introduced into all-susceptible large population, major outbreak occurs with around $95\%$ of the population and with high probability the epidemic is maximized to around $30\%$ of the population.
\end{abstract}

\section{Introduction}\label{intro}
\indent\par
According to WHO \cite{WEBSITE:1} report, as of 7:06pm CEST, 15 September 2022, there have been 607,745,726 confirmed cases of COVID-19, including 6,498,747 deaths, globally. The quick spread of this epidemic makes a worldwide impact on not only the human health but the economics and developments. Even worse, the uncontrollability of asymptomatic infections causes a problem that we may underestimate the transmission of the disease. \cite{ma2021global} estimates that the pooled percentage of asymptomatic infections among the confirmed cases was $40.50\%$ globally. Nevertheless, tough tactics, such as lockdown and traffic halt, to block the spread of the disease are not long plans because these require large social cost. Fortunately, vaccinations help the immunity and lower the probability of getting infected or dead from this disease \cite{WEBSITE:2}. As of 12 September 2022, a total of 12,613,484,608 vaccine doses, reported to WHO \cite{WEBSITE:1}, have been administered. Hence, it is important to consider the contribution of vaccinations to make more precise policy decision.

	 In \cite{hethcote2000mathematics} and \cite{brauer2019mathematical}, traditional epidemic models, such as SIS, SIR and SEIR models, and mathematical models of some well-known pandemics, such as smallpox, HIV and malaria, have been discussed in detail. In \cite{hsu2005modeling} and \cite{hsu2008role}, the authors considered asymptomatic infections and the quarantine policy when studying the transmission of SARS in 2003. \cite{de2020seiard}, \cite{chen2020time} and \cite{lin2020data} also investigate mathematical models of COVID-19 considering also the asymptomatic infections. Especially, \cite{chen2020time} and \cite{lin2020data} discretize the differential equations to difference equations and assume that some parameters are time-dependent. They use the data from different countries, such as China, the USA and Brazil, and apply a popular machine learning method: the Finite Impulse Response (FIR) filters to track the time series of these parameters and then predict the desired values, such as the number of infections and  recoveries. Furthermore, \cite{ghostine2021extended} considers vaccinations against COVID-19 and uses an Ensemble Kalman Filter to conduct data forecast. 
	 
	 In contrast of deterministic models, stochastic models are more realistic in epidemiology. Many stochastic models of epidemics are also proposed in recent years. The most common two ways of stochastic modeling of epidemics are to conduct a Markov branching process \cite{greenwood2009stochastic}, \cite{yuan2015optimal} and \cite{britton2019stochastic} and to conduct stochastic differential equations \cite{allen2017primer}, \cite{otunuga2017global} and \cite{rios2021studies}. The quantities in interest shall be probability distributions instead of exact values. For example, \cite{nakata2018mathcal} provides formulae of the final size in its deterministic model, while \cite{house2013big}, \cite{black2015computation} and \cite{icslier2020exact} give algorithms calculating the final size probability distribution.
	 
	The probability distributions of the final size and the maximum size are important topics in stochastic epidemic models. \cite{greenwood2009stochastic}, \cite{house2013big} and \cite{daniels1974maximum} are good materials for a review of calculations of these two distributions. Many optimized algorithms calculating these distributions have been proposed, such as  in \cite{black2015computation} and in \cite{icslier2020exact}. It was calculated that these probability distributions are bimodal.
	
	In this stage of the epidemic COVID-19, the quarantine policy is loosen gradually in many countries and so the data of quarantine is more unavailable. Hence, we will only consider an epidemic prevention policy: vaccination in our models in this paper. In this paper, we will conduct data forecasts to see whether the epidemic will ease or become more severe in the short term and how large the epidemic will be in the long run.
	
For the first question, we conduct a deterministic SAIVRD model using ordinary differential equations, which involves six compartments: susceptible people (S), (unconfirmed) asymptomatically infectious cases (A), (confirmed) isolated infectious cases (I), recovered people (R), vaccinated-immune people (V) and deaths (D). In this model, we will analyze the existence and the asymptotic stability of the disease-free and endemic equilibria related to the basic reproduction number $R_0$. It is shown that if $R_0>1,$ then there is a unique endemic equilibrium which is locally asymptotically stable, and if $R_0<1,$ then the disease-free equilibrium is locally asymptotically stable. We also assume the time-dependence of some parameters and track their time series to predict the time-varying parameters and values of confirmed cases $I$ using FIR filters. Using the data set \cite{WEBSITE:4}, \cite{WEBSITE:9} and \cite{WEBSITE:5} from the USA for $40$ days from June 30-th 2022 to August 8-th, 2022 to predict the next $10$ days, we get the relative prediction errors of $I$ within $3\%$ in $5$ days and $7\%$ in $10$ days, and the trend of change is
closed to the real data of $I$. This shows that our proposed model meets the real situation.

For the second question, we conduct a time-dependent Markov SARV model with no demography assuming that there will not be another outbreak and considering the administration of vaccinations. The four compartments are susceptible people (S), asmytomatic infections (A), isolations/recoveries/deaths (R) and vaccine-immune (V). Under this model, we approximate the probability distributions of the final size and the maximum size with one initial infectious person using the time-varying parameters obtained from the first forecast. 

The desired distributions are also bimodal in our model  with time-independent parameters in our simulations of smaller populations. For the model with time-dependent model, what are different are that there is another peak in the final size distribution, that the probability of ``minor outbreak'' \cite{allen2017primer} is higher with time-dependent parameters and that the maximum size distribution is oscillating. The distributions are similar between different populations and so we can estimate the “ratio” of the final size and the maximum size distributions by observing those in small populations. Under recent transmissibility of this disease in the USA, when an initial infection is introduced into all-susceptible large population, ``major outbreak'' \cite{allen2017primer} occurs with around $95\%$ of the population; with high probability the epidemic is maximized to around $30\%$ of the population. Moreover, to estimate the extinction probability of this epidemic in a large population, we assume the infinite population. As shown in \cite{icslier2020exact} and \cite{miller2018primer}, the plateau value of cumulative probability of the final size in a large population is approximately the extinction probability. It is also interesting to
note that this probability is approximately the reciprocal of the basic reproduction number $R_0^{\mathrm{SARV}}$ (defined in the SARV model) whenever $R_0^{\mathrm{SARV}}>1$,
which demonstrates numerically the result in \cite{allen2017primer}. 
These data forecasts can serve as reference to public health policy.

	We organize the article as follows. We derive our deterministic SAIVRD model, define the basic reproduction number $R_0$ and investigate the existence and asymptotic stability of the disease-free and endemic equilibria in \cref{secder}. \Cref{sectime} uses the FIR filter to train the time-varying rates to predict the trend of the epidemic in the short term. The probability distributions of the final size and the maximum size in the Markov SARV model with both time-independent and time-dependent parameters will be calculated in \cref{secfin}. \Cref{secnum} provides numerical results. \Cref{seccon} states our concluding remarks.
	
\section{The Derivation of the Deterministic SAIVRD Model and Stability Analysis}\label{secder}
	\indent\par
We firstly introduce our deterministic SAIVRD model in this section. The six compartments are defined in \cref{intro}. The model is under the following assumptions: there must be a period that an infected case is not confirmed yet which is asymptomatically infectious; asymptomatically infectious cases will not enter $R$ or $D$; a dead person cannot infect others; the system is closed, that is, there is no migration (the border is well-controlled); confirmed cases are well isolated so that they have very low adequate contact with others; nosocomial infections and re-infections are omitted.
 
The parameters are set as following and to be non-negative:

\begin{tabular}{|c|c|}
\hline
     Notation& Description \\
     \hline
     $B$& The new births per unit of time\\
     \hline
     $\mu$&Natural death rate\\
     \hline
     $\beta$&Transmission rate of $A$\\
     \hline
     $v$&Full-vaccination rate multiplying the vaccination efficiency\\
     \hline
     $w$&Progression rate from $A$ to $I$\\
     \hline
     $\rho$&Death rate from $I$\\
     \hline
     $\gamma$&Recovering rate from $I$\\
     \hline
\end{tabular}\\\\

Let $N_0$ be the initial total population.	The flow chart is shown as following.

\begin{tikzpicture}[node distance=1cm,
every node/.style={fill=white, font=\sffamily}, align=center]
\node[process] (S) {$S$};
\node[process, right of=S,xshift=2cm] (A) {$A$};
\node[process, left of=S,xshift=-1cm] (V) {$V$};
\node[process, right of=A,xshift=2cm] (I) {$I$};
\node[process, right of=I,xshift=2cm,yshift=1cm] (D) {$D$};
\node[process, right of=I,xshift=2cm,yshift=-1cm] (R) {$R$};
\draw[->]+(0pt,2cm)--node{$B$}(S);
\draw[->](S)--node{$\mu S$}+(0pt,-2cm);
\draw[->](V)--node{$\mu V$}+(0pt,-2cm);
\draw[->](A)--node{$\mu A$}+(0pt,-2cm);
\draw[->](I)--node{$\mu I$}+(0pt,-2cm);
\draw[->](R)--node{$\mu R$}+(0pt,-2cm);
\draw[->](S)--node{$vS$}(V);
\draw[->](S)--node{$\dfrac{\beta AS}{N_0}$}(A);
\draw[->](A)--node{$wA$}(I);
\draw[->](I)--node{$\rho I$}(D);
\draw[->](I)--node{$\gamma I$}(R);
\end{tikzpicture}

Then the model equations are \begin{equation}\label{eq:1}
	\begin{cases}
		S'(t)=B-vS(t)-\dfrac{\beta }{N_0}A(t)S(t)-\mu S(t)\\
		A'(t)=\dfrac{\beta }{N_0}A(t)S(t)-(w+\mu)A(t)\\
		I'(t)=wA(t)-(\rho+\gamma+\mu)I(t)\\
		V'(t)=vS(t)-\mu V(t)\\
		R'(t)=\gamma I(t)-\mu R(t)\\
		D'(t)=\rho I(t)
	\end{cases}
\end{equation} with the initial conditions $S(0)=S_0>0,A(0)=A_0>0,I(0)=V(0)=R(0)=D(0)=0.$ Clearly, we have the result in \cref{thm2.1}.

\begin{theorem}\label{thm2.1}
    All solutions of \cref{eq:1} always lie in the positively invariant set $$\Omega:=\{(S,A,I,V,R,D)\in\mathbb{R}^6:0\leq S+A+I+V+R+D\leq\frac{B}{\mu}+N_0\}$$
    for all $t\geq0$. In particular, $S$ and $A$ remain positive for all $t\geq0$.
\end{theorem}

\begin{proof}
    Note that $N=S+A+I+V+R+D$ and $N'=S'+A'+I'+V'+R'+D'=B-\mu N\mbox{ as functions of }t.$ Then $$0\leq N(t)=\frac{B}{\mu}+\big(N_0-\frac{B}{\mu}\big)e^{-\mu t}\leq\frac{B}{\mu}+N_0.$$Since $\displaystyle S'\geq -(v+\dfrac{\beta}{N_0} N(t)+\mu)S\geq-(v+\dfrac{\beta }{N_0}(\frac{B}{\mu}+N_0)+\mu)S,$ we have $$S\geq S_0\exp\big(-(v+\dfrac{\beta }{N_0}(\frac{B}{\mu}+N_0)+\mu)\big)>0.$$ Similarly for $A,I,V,R$ and $D$ and this theorem is proved.
\end{proof}

Since the number of deaths can be calculated by $D=N-S-A-I-V-R$, the model can be reduced to have five variables $S,A,I,V$ and $R$.

We now derive the basic reproduction number $R_0$ in our SAIVRD model.
The Jacobian matrix of the reduced model from \cref{eq:1} at an equilibrium $(S^*,A^*,I^*,V^*,R^*)$ is given by
$$
\begin{bmatrix}
-(v+\dfrac{\beta}{N_0} A^*+\mu)&-\dfrac{\beta}{N_0} S^*&0&0&0\\
\dfrac{\beta}{N_0} A^*&\dfrac{\beta }{N_0}S^*-(w+\mu)&0&0&0\\
0&w&-(\rho+\gamma+\mu)&0&0\\
v&0&\gamma&-\mu&0\\
0&0&0&0&-\mu
\end{bmatrix},$$whose characteristic polynomial is 
\begin{equation}\label{char}
    \begin{aligned}
     P(\lambda)=&(-(\rho+\gamma+\mu)-\lambda)(-\mu-\lambda)^2((v+\mu)(w+\mu-\dfrac{\beta}{N_0}S^*)\\
    &+\dfrac{\beta}{N_0} A^*(w+\mu)-(v+w+2\mu+\dfrac{\beta}{N_0}(A^*-S^*))\lambda+\lambda^2).
    \end{aligned}
\end{equation}

At the disease-free equilibrium (DFE) $(\displaystyle\frac{B}{v+\mu},0,0,\frac{v}{\mu}\frac{B}{v+\mu},0),$ the Jacobian of \cref{eq:1} is $$
\begin{bmatrix}
-(v+\mu)&\displaystyle-\frac{B\beta}{N_0(v+\mu)}&0&0&0\\
0&\displaystyle\frac{B\beta}{N_0(v+\mu)}-(w+\mu)&0&0&0\\
0&w&-(\rho+\gamma+\mu)&0&0\\
v&0&0&-\mu&0\\
0&0&\gamma&0&-\mu
\end{bmatrix},$$
whose eigenvalues are $-(v+\mu),-(\rho+\gamma+\mu),-\mu,-\mu,\displaystyle\frac{B\beta}{N_0(v+\mu)}-(w+\mu).$ Hence, we have $\dfrac{B\beta}{N_0(v+\mu)}-(w+\mu)<0$ if and only if the DFE is asymptotically stable. Hence, if we put \begin{equation}\label{r0}
    R_0=\frac{B\beta}{N_0(v+\mu)(w+\mu)},
\end{equation}then we have \cref{thm3.1}.
\begin{theorem}\label{thm3.1}
   The DFE $(\displaystyle\frac{B}{v+\mu},0,0,\frac{v}{\mu}\frac{B}{v+\mu},0)$ is locally asymptotically stable if and only if $R_0<1.$
\end{theorem}

We then define the basic reproduction number by \cref{r0} and explain why this definition of $R_0$ is reasonable in epidemiology as follows. Firstly, the value of $R_0$ has positive correlations with $B$ and $\beta$, which means that there are more people will be infected when more people are born 
    or when more people are contacted by infectious people. Secondly, $R_0$ has negative correlations
    with $w,v$ and $\mu$, which means that if we can improve the efficiency of disease screening (to shorten the period of infecting others), improve the efficiency of vaccinations or the infected people die fast, then fewer people would be infected.

It is also worthwhile to note that the rates of recovery $\gamma$ and death $\rho$ of the confirmed cases are not involved in $R_0$. This is because the class $I$ is well-controlled such that people in the class $I$ cannot cause further infections.

There may be an equilibrium in which $A^*$ and $I^*$ are positive in \cref{eq:1}, whose stability means that the human will coexist the virus finally and which is called an endemic equilibrium. The existence and asymptotic stability of the endemic equilibria are stated in \cref{thm3.3}.
\begin{theorem}\label{thm3.3}
	There is a unique locally asymptotically stable endemic equilibrium whenever $R_0>1$.
\end{theorem}
\begin{proof}
	Solving	$$\begin{cases}
	0=B-vS^*-\dfrac{\beta}{N_0}A^*S^*-\mu S^*\\
	0=\dfrac{\beta}{N_0} A^*S^*-(w+\mu)A^*\\
	0=wA^*-(\rho+\gamma+\mu)I^*\\
	0=vS^*-\mu V^*\\
	0=\gamma I^*-\mu R^*,
	\end{cases}$$ we have $$A^*=\frac{\rho+\gamma+\mu}{w}I^*,R^*=\frac{\gamma}{\mu}I^*,S^*=\frac{w+\mu}{\beta}N_0,V^*=\frac{v(w+\mu)}{\beta\mu}N_0,$$where $I^*=\dfrac{B-(v+\mu)(w+\mu)N_0/\beta}{(1+1/w)(w+\mu)(\rho+\gamma+\mu)}.$ Since $I^*>0$ if and only if $R_0>1,$ the existence follows.\\
    Plug this equilibrium into \cref{char} and then under the assumptions, since\\ $\dfrac{\beta}{N_0} A^*(w+\mu)>0,$ the roots of $P(\lambda)$ are all negative if and only if $v+w+2\mu+\dfrac{\beta}{N_0}(A^*-S^*)>0,$ which can be easily calculated to be equivalent to $R_0+w>0$ and so this theorem is proved. 
\end{proof}

\section{Prediction Methods Using Time-Varying Parametric Models}\label{sectime}
\indent\par
In this section, we assume the time-dependence of some parameters in our SAIVRD model derived in \cref{secder} to conduct the predictions of infections and time-varying rates in the short term. Since the data is updated in days, it is reasonable to consider the discrete-time model instead of the ordinary differential equations.

Suppose that the parameters $B$ and $\mu$ are constants and that the other parameters $\rho(t),\gamma(t),w(t),v(t)$ and $\beta(t)$ are time-dependent. Then we have the model difference equations in \cref{eq5}.

\begin{equation}\label{eq5}
	\begin{cases}
		S(t+1)-S(t)=B-(v(t)+\mu)S(t)-\dfrac{\beta(t)}{N_0} A(t)S(t)\\
		A(t+1)-A(t)=\dfrac{\beta(t)}{N_0} A(t)S(t)-(w(t)+\mu)A(t)\\
		I(t+1)-I(t)=w(t)A(t)-(\rho(t)+\gamma(t)+\mu)I(t)\\
		V(t+1)-V(t)=v(t)S(t)-\mu V(t)\\
		R(t+1)-R(t)=\gamma(t)I(t)-\mu R(t)\\
		D(t+1)-D(t)=\rho(t) I(t).
	\end{cases}
\end{equation}

We will use the known data $S(t),A(t),I(t),V(t),R(t)$ and $D(t)$ for $t=0,1,\cdots,\\T-1$ to track these five time series: $\rho(t),\gamma(t),w(t),v(t)$ and $\beta(t).$ Then we predict the values of $S(T),A(T),I(T),V(T),R(T),D(T),\rho(T-1),\gamma(T-1),w(T-1),v(T-1)$ and $\beta(T-1).$

For $t=0,1,\cdots,T-2,$ we solve $\rho(t),\gamma(t),w(t),v(t)$ and $\beta(t)$ in \cref{eq6}.
\begin{equation}\label{eq6}
\begin{cases}
  \rho(t)=\dfrac{D(t+1)-D(t)}{I(t)},\\
\gamma(t)=\dfrac{R(t+1)-(1-\mu)R(t)}{I(t)},\\
w(t)=\dfrac{I(t+1)+(\rho(t)+\gamma(t)+\mu-1)I(t)}{A(t)},\\
v(t)=\dfrac{V(t+1)-(1-\mu)V(t)}{S(t)},\\ \beta(t)=\dfrac{A(t+1)+(w(t)+\mu-1)A(t)}{A(t)S(t)}\cdot N_0
,\mbox{ for }t=0,1,\cdots,T-2. 
\end{cases}
\end{equation}

In our numerical experiment in \cref{secnum}, some $\beta(t)$ are negative, which is epidemiologically unreasonable. In this case, we reset $\beta(t)$ to be $0$. 

We will apply a data set in the USA from \cite{WEBSITE:4}, \cite{WEBSITE:9} and \cite{WEBSITE:5} that is of large number of infections and so the Finite Impulse Response (FIR) filters is an appropriate method to predict the time-varying rates \cite{chen2020time}:  
 \begin{equation}\label{fir}
 \begin{cases}
  \hat{\rho}(t)=x_1\rho(t-1)+\dots+x_{J_1}\rho(t-J_1)+x_0=\sum_{j=1}^{J_1}x_j\rho(t-j)+x_0\\
   \hat{\gamma}(t)=y_1\gamma(t-1)+\dots+y_{J_2}\gamma(t-J_2)+y_0=\sum_{j=1}^{J_2}y_j\gamma(t-j)+y_0\\
     \hat{w}(t)=z_1w(t-1)+\dots+z_{J_3}w(t-J_3)+z_0=\sum_{j=1}^{J_3}z_jw(t-j)+z_0\\
  \hat{v}(t)=u_1v(t-1)+\dots+u_{J_4}v(t-J_4)+u_0=\sum_{j=1}^{J_4}u_jv(t-j)+u_0\\
 \hat{\beta}(t)=b_1\beta(t-1)+\dots+b_{J_5}\beta(t-J_5)+b_0=\sum_{j=1}^{J_5}b_j\beta(t-j)+b_0,
 \end{cases}
 \end{equation}
 where $J_i, i=1,2,3,4,5$ are orders of the five FIR filters and $x_j,y_j,z_j,u_j,b_j$ are coefficients of the impulse responses of these five FIR filters. The coefficients can be determined by considering the following minimization problems.
 \begin{equation}\label{min}
 \begin{cases}
\min\limits_{{\bf{x}}\in\mathbb{R}^{J_1+1}}\sum_{t=J_1}^{T-2}(\rho(t)-\hat{\rho}(t))^2+\alpha_1\sum_{j=0}^{J_1}x_j^2\\ 
\min\limits_{{\bf{y}}\in\mathbb{R}^{J_2+1}} \sum_{t=J_2}^{T-2}(\gamma(t)-\hat{\gamma}(t))^2+\alpha_2\sum_{j=0}^{J_2}y_j^2\\
\min\limits_{{\bf{z}}\in\mathbb{R}^{J_3+1}}\sum_{t=J_3}^{T-2}(w(t)-\hat{w}(t))^2+\alpha_3\sum_{j=0}^{J_3}z_j^2\\
   \min\limits_{{\bf{u
 }}\in\mathbb{R}^{J_4+1}}\sum_{t=J_4}^{T-2}(v(t)-\hat{v}(t))^2+\alpha_4\sum_{j=0}^{J_4}u_j^2\\
 \min\limits_{{\bf{b
 }}\in\mathbb{R}^{J_5+1}}\sum_{t=J_5}^{T-2}(\beta(t)-\hat{\beta}(t))^2+\alpha_5\sum_{j=0}^{J_5}b_j^2
 \end{cases}
 \end{equation}
 using rigid regression, where $\alpha_i,i=1,2,3,4,5$ are regulation parameters, which avoid overfitting. After predicting the values of $\hat{\rho}(T-1),\hat{\gamma}(T-1),\hat{w}(T-1),\hat{v}(T-1)$ and $\hat{\beta}(T-1),$ we plug these data back into \cref{eq5} to get the desired values $\hat{S}(T),\hat{A}(T),$ $\hat{I}(T),\hat{V}(T),\hat{R}(T)$ and $\hat{D}(T)$ by replacing $t$ with $T-1$.
 
After doing prediction for one day, we can treat this predicted day as known data to predict more one day, and step by step we calculate more days. More precisely, we will predict the values of $\hat{\rho}(t),\hat{\gamma}(t),\hat{w}(t),\hat{v}(t),\hat{\beta}(t),\hat{S}(t+1),\hat{A}(t+1),\hat{I}(t+1),\hat{V}(t+1),\hat{R}(t+1)$ and $\hat{D}(t+1)$ for $t=T-1,T,\cdots,T+d-1.$

In this process, we still need to determine the orders $J_i,i=1,2,3,4,5$ and find the formulae to solve the optimization problems in \cref{min}.
Given the orders, we can use a result from \cite{lin2020data} named ``Normal Gradient Equation'' and described as follows, to compute the desired coefficients.
  Let $f(0),f(1),\cdots,f(T-2)$ be the known data, and the FIR filter  $\hat{f}(t)=\sum_{j=1}^{J}x_jf(t-j)+x_0$ be the predicted data, where $t=J,\cdots,T-2$. Define the function $F$ by $$F(x_0,x_1,\dots,x_J;\alpha)=\sum_{t=J}^{T-2}(f(t)-\hat{f}(t))^2+\alpha\sum_{j=0}^{J}x_j^2,$$
  where $\alpha$ is the regulation parameter.
 	If $\det(\alpha I+A)\neq 0,$ then $$(\alpha I+A)^{-1}{\bf{b}}=\operatornamewithlimits{argmin}\limits_{{\bf{x}}\in\mathbb{R}^{J+1}}F({\bf{x}},\alpha),$$ where $I$  is $(J+1)\times (J+1)$ identity matrix,\\
 	$A=\begin{bmatrix}
 	(T-J-1) & \sum_{t=J}^{T-2}f(t-1) & \cdots & \sum_{t=J}^{T-2}f(t-J)\\
 	\sum_{t=J}^{T-2}f(t-1) & \sum_{t=J}^{T-2}(f(t-1))^2 & \cdots & \sum_{t=J}^{T-2}f(t-1)f(t-J)\\
 	\vdots & \vdots & \ddots & \vdots \\
 	\sum_{t=J}^{T-2}f(t-J) & \sum_{t=J}^{T-2}f(t-J)f(t-1) & \cdots & \sum_{t=J}^{T-2}(f(t-J))^2
 	\end{bmatrix}$ and\\
 	${\bf{b}}=\begin{bmatrix}
 	\sum_{t=J}^{T-2}f(t)\\
 	\sum_{t=J}^{T-2}f(t)f(t-1)\\
 	\vdots \\
 	\sum_{t=J}^{T-2}f(t)f(t-J)
 	\end{bmatrix}$.

We give the process of calculation in \cref{algorder} to \cref{algsev} \cite{lin2020data}. Firstly, we use \cref{algorder} to find the best choice of orders. Then after some pre-processing of data, we use \cref{algnext} to predict the next-day data of $S,A,I,V,R,D$ and the time-varying rates. Finally, we predict the data for several days in \cref{algsev}. 

The procedure of \cref{algorder} is described as follows. Firstly, we input the known data (the time-dependent parameters), and decide how large $\ell_T$ the training set we want to extract from is, bounds $U_J$ and $L_J$ and regulation parameters. Then we use different $J$'s to find the best $J$ which gives the least residual.

\begin{algorithm}
\caption{Best Order Searcher}
\label{algorder}
		\textbf{Input}: Data $=(f(0),f(1),\cdots,f(T-2))$; training size: $\ell_T$; lower bound of order: $L_J$ ; upper bound of order: $U_J$; Regulation Parameter: $\alpha$.\\
		Define $\mathrm{Data}_{T}=(f(0),f(1),\cdots, f(\ell_T-1))$ and $\mathrm{Data}_V=(f(\ell_T),f(\ell_T+1),\cdots,f(T-2))$.\\
		Calculate the validation length: $\ell_V\leftarrow T-\ell_T-1$.
		\begin{algorithmic}
			\For{$J\leftarrow L_J$ \textbf{to} $U_J$}
			\For{$t\leftarrow\ell_T$ \textbf{to} $T-2$}
			\State Calculate ${\bf{x}}=(x_0,x_1,\cdots,x_J)^T=\operatornamewithlimits{argmin}\limits_{{\bf{x}}\in\mathbb{R}^{J+1}}F({\bf{x}},\alpha)$
			\State Estimate $\hat{f}_J(t)=\displaystyle\sum_{j=1}^Jx_jf(t-j)+x_0,$ where $f(t-1),f(t-2),\cdots,f(\ell_T)$ are taken to be $\hat{f}_J(t-1),\hat{f}_J(t-2),\cdots,\hat{f}_J(\ell_T)$ and $f(\ell_T-1),f(\ell_T-2),\cdots,f(t-J+1),f(t-J)$ are taken from $\mathrm{Data}_T$.
			\EndFor
   \State Calculate $\mathrm{err}(J):=\displaystyle\sum_{t=\ell_T}^{T-2}|\hat{f}_J(t)-f(t)|.$
   \EndFor 
		\end{algorithmic}
		\textbf{return} Order of FIR: $J_{\mathrm{fit}}=\operatornamewithlimits{argmin}\limits_{J=L_J,\cdots,U_J}\mathrm{err}(J)$.
\end{algorithm}

We next conduct short-term prediction in \cref{algnext}. Due to the limitation of available data from \cite{WEBSITE:5}, we need to make more modifications and assumptions as follows. Firstly, the cumulative confirmed cases $\tilde{I}$ and cumulative vaccinations $\tilde{V}$ provided in our data set \cite{WEBSITE:5} is different from the definition of $I$ and $V$ in our model. Our definition of $I$ is the simultaneous confirmed cases, and $V$ is the number of valid immunity by vaccination. To adapt our definition, we update $I$ to be $(\tilde{I}-R-D)(1-\mu)$ and update $V$ to be $\tilde{V}\sigma(1-\mu),$ where $\sigma$ is the efficiency of vaccination. Secondly, the asymptotic infections $A$ is updated by $A=pI,$ where $p$ is the estimated ratio between asymptomatic and symptomatic infections. Thirdly, $S(t)$ is updated by $S(t)=(N_0-A(t)-I(t)-V(t)-R(t)-D(t)+tB)(1-\mu)$ for $t=0,1,\cdots,T-1.$

After these settings, we can predict the data for one day by \cref{algnext}. Firstly, we calculate the time-varying rate by known data and reset $\beta(t)$ if it is negative. Then use \cref{algorder} to find the best FIR orders for prediction and calculate $\operatornamewithlimits{argmin}\limits_{{\bf{x}}\in\mathbb{R}^{J+1}}F({\bf{x}},\alpha)$ to predict the rates. Finally, plug these predicted rates into \cref{eq5} to predict the data of the six compartments.

\begin{algorithm}
\caption{Short-Term Prediction}
\label{algnext}
\textbf{Input}: Initial total population $N_0$, number $T$ of days of known data, revised numbers $S(t),A(t),I(t),V(t),R(t)$ and $D(t)$ for $t=0,1,\cdots,T-1,$ and regulation parameters $\alpha_i,i=1,2,3,4,5.$
	\begin{algorithmic}
\For{$t\leftarrow 0$ \textbf{to} $T-2$}
\State Calculate $\rho(t),\gamma(t),w(t),v(t)$ and $\beta(t)$ by \cref{eq6}.
 \If{$\beta(t)<0$}
        \State $\beta(t)\leftarrow0.$
    \EndIf
\EndFor
	
Find the best FIR orders $J_i,i=1,2,3,4,5$ by \cref{algorder}.\\
Calculate the coefficients by $\operatornamewithlimits{argmin}\limits_{{\bf{x}}\in\mathbb{R}^{J+1}}F({\bf{x}},\alpha)$ using the best orders.\\
 Calculate $\hat{\rho}(T-1),\hat{\gamma}(T-1),\hat{w}(T-1),\hat{v}(T-1)$ and $\hat{\beta}(T-1)$ by \cref{fir} and $\hat{S}(T),\hat{A}(T),\hat{I}(T),\hat{V}(T),\hat{R}(T)$ and $\hat{D}(T)$ by \cref{eq5}.
  \If{$\hat{\beta}(T-1)<0$}
        \State $\hat{\beta}(T-1)\leftarrow0.$
    \EndIf\end{algorithmic}
\textbf{return} $\hat{\rho}(T-1),\hat{\gamma}(T-1),\hat{w}(T-1),\hat{v}(T-1),\hat{\beta}(T-1),\hat{S}(T),\hat{A}(T),\hat{I}(T),\hat{V}(T),\hat{R}(T)$ and $\hat{D}(T)$.
\end{algorithm}

When we get the prediction for one day, we can treat them as the known data to predict more days by \cref{algsev}. Firstly, apply \cref{algnext} to predict the data for one day. Then iteratively treat the predicted data as known data to conduct predictions for several days and our goal of this section is accomplished so far.

\begin{algorithm}[h!]
\caption{Prediction for Several Days}
\label{algsev}
\textbf{Input}: Same as in \cref{algnext} and Number $d$ of predicted days in addition.\\
Calculate $\hat{\rho}(T-1),\hat{\gamma}(T-1),\hat{w}(T-1),\hat{v}(T-1)$ and $\hat{\beta}(T-1)$, and $\hat{S}(T),\hat{A}(T),\hat{I}(T),\hat{V}(T),\hat{R}(T)$ and $\hat{D}(T)$ by \cref{algnext}.
\begin{algorithmic}
\For{$t\leftarrow T$ \textbf{to} $T+d-1$}
\State Redefine $I$ by concatenating $I$ and $\hat{I}(t)$ and $V,R,D$ are redefined in the similar way.
\State Calculate $\hat{\rho}(t),\hat{\gamma}(t),\hat{w}(t),\hat{v}(t)$ and $\hat{\beta}(t),$ and $\hat{S}(t+1),\hat{A}(t+1),\hat{I}(t+1),\hat{V}(t+1),\hat{R}(t+1)$ and $\hat{D}(t+1)$ by \cref{algnext}.
\EndFor
\end{algorithmic}
\textbf{return} $\hat{\rho}(t),\hat{\gamma}(t),\hat{w}(t),\hat{v}(t)$ and $\hat{\beta}(t),$ and $\hat{S}(t+1),\hat{A}(t+1),\hat{I}(t+1),\hat{V}(t+1),\hat{R}(t+1)$ and $\hat{D}(t+1)$ for $t=T-1,T,\cdots,T+d-1.$
\end{algorithm}

\section{Probability Distributions of the Final Size and the Maximum Size}\label{secfin}\indent\par
To estimate how large the epidemic caused by a single infected person is in the long run, we calculate the probability distributions of the final size and the maximum size. The final size is defined to be the cumulative number of the infections until the end of epidemic and the maximum size is defined to be the maximum number of the simultaneous infections during the prevalence of the epidemic. We propose a simplified stochastic SARV model whose four compartments are defined in \cref{intro}. We do not consider the demography and assume that there will not be another outbreak, that 
each asymptomatic infectious case transmits the disease to others according to a Poisson process with parameter $\lambda$, that when a susceptible is infected, it enters $A$, that $A$ enters $R$ (isolated, dead or recovered) after the period $T$ exponentially distributed with parameter $\alpha$ and that $\theta$ is the average ratio of the number of vaccinations to the number of jumps \cite{greenwood2009stochastic}.

The reasons why we treat the isolated cases as the compartment $R$ are that they cannot infect others and that this model does not require huge calculations.

Let $N$ be the total population. Then the process $\{(a(j),s(j))\}$ is a discrete-time Markov process on the state space $$X=\{(a,s):a+s\leq N\}.$$

The flow chart is shown as following.
\begin{center}
\begin{tikzpicture}[node distance=0.6cm,
every node/.style={fill=white, font=\sffamily}, align=center]
\node[process] (S) {$S$};
\node[process, right of=S,xshift=2cm] (A) {$A$};
\node[process, left of=S,xshift=-2cm] (V) {$V$};
\node[process, right of=A,xshift=2cm] (R) {$R$};
\draw[->](S)--node{$\dfrac{\lambda AS}{N}$}(A);
\draw[->](S)--node{$\theta S$}(V);
\draw[->](A)--node{$\alpha A$}(R);
\end{tikzpicture}
\end{center}

The basic reproduction number is then defined to be $$R_0^{\mathrm{SARV}}=\dfrac{\lambda(1-\theta)}{\alpha}.$$

The parameters under these assumptions are time-independent. However, just as in \cref{sectime}, the parameters can also be time-dependent. In \cref{final}, we will calculate the probability distributions of the final size and the maximum size with time-independent parameters, whose calculation is based on \cite{greenwood2009stochastic} and \cite{icslier2020exact}. In \cref{secmartime}, we assume the time-dependence of parameters and apply the parameters from the first forecast in \cref{sectime} to approximate these two distributions and compare them with time-independent model.

\subsection{Distributions with Time-Independent Parameters}\label{final}\indent\par
Firstly, we can calculate the transition probabilities as following:
$$\mathbb{P}((A(j+1),S(j+1))=(a+1,s-1)|(A(j),S(j))=(a,s))=\dfrac{\lambda s(1-j\theta)}{\lambda s(1-j\theta)+N\alpha},$$
$$\mathbb{P}((A(j+1),S(j+1))=(a-1,s)|(A(j),S(j))=(a,s))=\dfrac{N\alpha}{\lambda s(1-j\theta)+N\alpha}.$$

Let $\tau=\inf\{j\geq0:A(j)=0\}$
and $$P_m(a,s)=\mathbb{P}((A(\tau),S(\tau))=(0,N-m)|(A(0),S(0))=(a,s)),$$which gives the probability that the final cumulative number of infections is $m$ with initial state $(a,s).$ Then what we seek is the final size distribution $P_m(1,N-1)$ for $m=0,1,\cdots,N.$

We firstly set up the boundary conditions in \cref{eq7}:
\begin{equation}\label{eq7}
\begin{cases}
    P_m(a,s)=0 \mbox{ for } s=0,1,\cdots,N-m-1, \mbox{ and } a=0,1,\cdots,N-s.\\
    P_m(0,N-m)=1.
\end{cases}
\end{equation}

The concept behind these is very simple: if the number of susceptible people is less than $N-m,$ the cumulative number of infections must be greater than $m$ and so with probability $0$ the epidemic ends with size $m.$ If there is no infectious people and $N-m$ susceptible people, then the epidemic must stop with size $m.$

For $s=1,2,\cdots,N$ and $a=1,2,\cdots,N-s$, by conditioning on the first step \cite{durrett1999essentials}, we have
\begin{equation}\label{eq8}
    P_m(a,s)=\dfrac{\lambda s(1-\theta)}{\lambda s(1-\theta)+N\alpha}P_m(a+1,s-1)+\dfrac{N\alpha}{\lambda s(1-\theta)+N\alpha}P_m(a-1,s).
\end{equation}
The iteration can be explained as following: $\dfrac{\lambda s(1-\theta)}{\lambda s(1-\theta)+N\alpha}$ and $\dfrac{N\alpha}{\lambda s(1-\theta)+N\alpha}$ are the probabilities of infection and recovery in a step, respectively. The former multiplying $P_m(a+1,s-1)$ is the probability that when infection occurs in the first step, the epidemic ends with size $m;$ the latter multiplying $P_m(a-1,s)$ is the probability that when recovery occurs in the first step, the epidemic ends with size $m.$ $P_m(a,s),$ the probability that the cumulative number of infections is $m$ with initial state $(a,s),$ equals the sum of these two probabilities.

Similarly as the final size distribution, we can calculate the maximum size distribution.
 Let $$Q_m(a,s)=\mathbb{P}(\sup\limits_{j\in[0,\tau]}A(j)=m|(A(0),S(0))=(a,s))$$ for $m=1,2,\cdots,N.$ Then what we seek is the maximum size distribution $Q_m(1,N-1)$ for $m=1,2,\cdots,N$ and we have the boundary conditions and the iteration formula as in \cref{eq9}.
\begin{equation}\label{eq9}
    \begin{cases}
        Q_m(m,0)=1.\\
        Q_m(a,0)=0 \mbox{ for } a\neq m.\\
        Q_0(0,s)=1.\\
        Q_m(0,s)=0 \mbox{ for } m>0.\\
        Q_m(a,s)=0 \mbox{ for } m<a \mbox{ or } a+s<m.\\
        Q_m(a,s)=\dfrac{\lambda s(1-\theta)}{\lambda s(1-\theta)+N\alpha}Q_m(a+1,s-1)\\\hspace*{12em}+\dfrac{N\alpha}{\lambda s(1-\theta)+N\alpha}Q_m(a-1,s),\mbox{  otherwise.}
    \end{cases}
\end{equation}

The iteration can be explained in a similar way of that of the final size distribution. Since when $s=0$ or $a=0,$ there will no more infection occurs. It follows directly by the first to fourth conditions. If $m<a,$ then the initial state contains the epidemic size larger than $m;$ if $a+s<m,$ then the epidemic size must be no more than $m$ afterwards and so the fifth condition follows.

\subsection{Distributions with Time-Varying Parameters}\label{secmartime}
\indent\par As in \cref{sectime}, the parameters can be time-dependent, so we can use the predicted time-varying rates to calculate the distributions of final size and maximum size. To adapt the simplified model, let
\begin{equation}\label{eqtilde}
    \tilde{\lambda}(t)=\beta(t),\tilde{\alpha}(t)=w(t)\mbox{ and }\tilde{\theta}(t)=v(t),
\end{equation} where $t$ is the count of predicted days. Let $j$ be the step in which $(a,s)$ lies in starting from $(1,N-1).$ 

It is not reasonable to assume that the process proceeds one step in one day, so we assume that there are $r$ steps in one day. That is, there are $r$ $j$'s corresponding to a same $t$. 
Then the step $j$ lies in the $t_j:=\lfloor\dfrac{j}{r}\rfloor$-th day and so we can set \begin{equation}\label{eqtj}
    \lambda(j)=\dfrac{\tilde{\lambda}(t_j)}{r},\alpha(j)=\dfrac{\tilde{\alpha}(t_j)}{r}\mbox{ and }\theta(j)=\dfrac{\tilde{\theta}(t_j)}{r}.
\end{equation} 

With the same boundary conditions as in \cref{eq7}, we rewrite \cref{eq8} as \begin{equation}\label{eq10}
\begin{aligned}
     P_m(a,s)=&\dfrac{\lambda(j)s(1-\theta(j))}{\lambda(j)s(1-\theta(j))+N\alpha(j)}P_m(a+1,s-1)\\ &+\dfrac{N\alpha(j)}{\lambda(j) s(1-\theta(j))+N\alpha(j)}P_m(a-1,s).
\end{aligned}
\end{equation}

The choice of $r$ must satisfy that $t_j\leq T,$ the total number of predicted days.
\par Note that every step moves a state by either $(-1,0)$ or $(1,-1).$ For a state $(a,s)$ with $s=1,2,\cdots,N$ and $a=1,2,\cdots,N-s,$ write $(a,s)=(1,N-1)+(x,-y),$ that is, $(1,N-1)$ is moved to $(a,s)$ by $(x,-y).$ In the $j$-th step, $(x,-y)=(-j+2y,-y).$ Then $j=2y-x=2(N-1-s)-(a-1).$ 

Similarly for the maximum size distribution, we set the same boundary conditions as in \cref{eq9} and rewrite the iteration as \begin{equation}\label{eq11}\begin{aligned}
     Q_m(a,s)=&\dfrac{\lambda(j)s(1-\theta(j))}{\lambda(j)s(1-\theta(j))+N\alpha(j)}Q_m(a+1,s-1)\\ &+\dfrac{N\alpha(j)}{\lambda(j) s(1-\theta(j))+N\alpha(j)}Q_m(a-1,s).
\end{aligned}
\end{equation}

\Cref{algfinal} computes 
the final size and the maximum size distributions with time-dependent parameters. The first and third for loops set up the boundary conditions of $P_m$ and $Q_m.$ Then we compute $P_m$ iteratively in the second for loop. The fourth for loop computes $Q_m$ iteratively \cite{icslier2020exact}.
\begin{algorithm}
\caption{Final Size and Maximum Size Distributions with Time-Dependent Parameters}
\label{algfinal}
\textbf{Input}: Total population $N,$ the number $r$ of steps in one day and predicted parameters $\tilde{\lambda}(t),\tilde{\alpha}(t)$ and $\tilde{\theta}(t)$ for $t=0,1,\cdots,T-1.$
\begin{algorithmic}
\For{$m\leftarrow0$ \textbf{to} $N$}
\State $P_m=O_{(N+1)\times(N+1)}.$
\State $P_m(0,N-m)=1.$
\State $Q_m=O_{(N+1)\times(N+1)}.$
\EndFor
\For{$s\leftarrow1$ \textbf{to} $N$}
\For{$a\leftarrow1$ \textbf{to} $N-s$}
\State $t_j\leftarrow\lfloor\dfrac{2(N-1-s)-(a-1)}{r}\rfloor.$
\For{$m\leftarrow0$ \textbf{to} $N$}
\State Calculate $P_m(a,s)$ by \cref{eqtilde}, \cref{eqtj} and \cref{eq10}.
\EndFor
\EndFor
\EndFor
\For{$a\leftarrow0$ \textbf{to} $N$}
\State $Q_a(a,0)=1.$
\EndFor
\For{$s\leftarrow1$ \textbf{to} $N-1$}
\State $Q_0(0,s)=1.$
\For{$a\leftarrow1$ \textbf{to} $N-s$}
\State $t_j\leftarrow\lfloor\dfrac{2(N-1-s)-(a-1)}{r}\rfloor.$
\For{$m\leftarrow a+1$ \textbf{to} $a+s$}
\State Calculate $Q_m(a,s)$ by \cref{eqtilde}, \cref{eqtj} and \cref{eq11}.
\EndFor
\State $Q_a(a,s)=1-\sum_{m=a+1}^{a+s}Q_m(a,s).$
\EndFor
\EndFor
\end{algorithmic}
\textbf{return} $P=(P_0(1,N-1),P_1(1,N-1),\cdots,P_N(1,N-1))$ and $Q=(Q_0(1,N-1),Q_1(1,N-1),\cdots,Q_N(1,N-1))$.
\end{algorithm}

\subsection{Extinction Probability in a Large Population}\indent\par
In this subsection, we discuss the extinction probability with one initial
infection in an infinite (or very large) population with exponential distributed infectious period related to the basic reproduction number $R^{\mathrm{SARV}}_0$ in our stochastic SARV model. We will also see the relation between the extinction probability and the final size distribution.

The following calculation is based on the method proposed in \cite{andersson2012stochastic}. \cite{icslier2020exact} also did this discussion in their proposed $SIkR$ model to approximate the outbreak probability, which is complementary to the extinction probability.

Let $X$ be a random variable of number of infections caused by a single infectious person during exponential distributed infectious period. Let $T$ be the random variable of the infectious period. Since $\dfrac{\lambda Ts(1-\theta)}{N}=\dfrac{\lambda T(N-a)(1-\theta)}{N}\rightarrow\lambda T(1-\theta)$ as $N\to\infty,$ we write the probability generating function of $X$ as $\mathbb{E}(u^X)=\mathbb{E}(e^{-\lambda T(1-\theta)(1-u)}).$ Using the moment generating function $\mathbb{E}(e^{tT})=\dfrac{\alpha}{\alpha-t}$ of $T,$ we have $$\mathbb{E}(u^X)=\dfrac{\alpha}{\alpha+\lambda(1-\theta)(1-u)}.$$ Then the extinction probability $u$ of the Markov branching process is the smaller solution of the fixed point equation $\dfrac{\alpha}{\alpha+\lambda(1-\theta)(1-u)}=u,$ which gives $$u=\min\{1,\dfrac{\alpha}{\lambda(1-\theta)}\}=\min\{1,\dfrac{1}{R_0^{\mathrm{SARV}}}\}.$$

For the relation with the final size distribution, we do the following calculation considering the method proposed by \cite{miller2018primer}. Let $\Omega_g(z)=\displaystyle\sum_{j}\omega_j(g)z^j$ be the probability generating function for the distribution of the number of recoveries at step $g$. Let $\Omega_\infty(z)=\lim\limits_{g\rightarrow\infty}\Omega_g(z)$ be the pointwise limit. Let $\omega_\infty$ be the probability that the cumulative number of the infections is infinite and define $z^\infty=1$ as $z=1$ and $z^\infty=0$ as $z\in[0,1)$. Then we can express $\Omega_\infty(z)$ by $\Omega_\infty(z)=\displaystyle\sum_{r<\infty}\omega_rz^r+\omega_\infty z^\infty.$ Let $z\to 1-$ and we have $\displaystyle\sum_{r<\infty}\omega_r=1-\omega_\infty,$ which gives the extinction probability. Therefore, this probability is exactly the cumulative probability of the final size which is finite in an infinite population. In order to demonstrate this concept graphically, we will use a small population $N=1000$ with time-independent parameters in \cref{dis} to observe the plateau value in the cumulative probability distribution of the final size, which is approximately the value of the extinction probability $u$.

\section{Numerical Results}\label{secnum}
\subsection{Data Setting}\label{data}\indent\par
For the first forecast, we track the data for $40$ days in the USA from \cite{WEBSITE:4}, \cite{WEBSITE:9} and \cite{WEBSITE:5} from June 30-th, 2022 to August 8-th, 2022 to predict the data on next 10 days, i.e., August 9-th, 2022 to August 18-th, 2022, and estimate the relative errors.

$L_J$ is taken to be $2$. $U_J$ and $\ell_T$ are taken to be $35$. $\sigma$ is taken (underestimate) to be $0.5$ \cite{WEBSITE:3} and $p$ is taken to be $\dfrac{40}{100-40}=\dfrac{2}{3}$ \cite{ma2021global}. For demography, it was estimated in \cite{WEBSITE:4} that the total population in the USA is $N_0=338,279,857$, that the number of new births per day in the USA is about $B=10,800$ and that death rate per day is $\mu=\dfrac{7855}{N_0}$.

Unlike the FIR filters that can predict the real-world quantities, the calculation of the desired probability distributions is too large to apply the real data. Hence, we only apply the estimated rates from the short-term prediction for $20$ days to small populations $N=100,N=500$ and $N=1000$ to demonstrate the use of calculation of the desired distributions. Finally, we consider the rates to be time-dependent. We assume that the epidemic ends within $20$ days and set $r=10,r=50$ and $r=100,$ respectively. We will see that the distributions are similar between different populations and so we can expect that with the same parameters and in a large population, the ``ratios'' of the final size and the maximum size are distributed similarly.

\subsection{Prediction Using Time-Varying Parametric Models}\label{next}
\indent\par The predicted time-varying $v(t)$ is around $1.2423\times 10^{-4},\rho(t)$ is around $1.1445\times 10^{-4}$ and $\gamma(t)$ is around $0.0305.$ The predicted values of $\beta(t)$ and $w(t)$ vary drastically and they are presented in \cref{fig:beta}; 
\cref{fig:I} shows the real and predicted data of confirmed $I(t)$; 
the time-varying basic reproduction numbers from in the SAIVRD model are presented in \cref{fig:R0}.
\begin{figure}
    \centering
    \includegraphics[height=3.5cm]{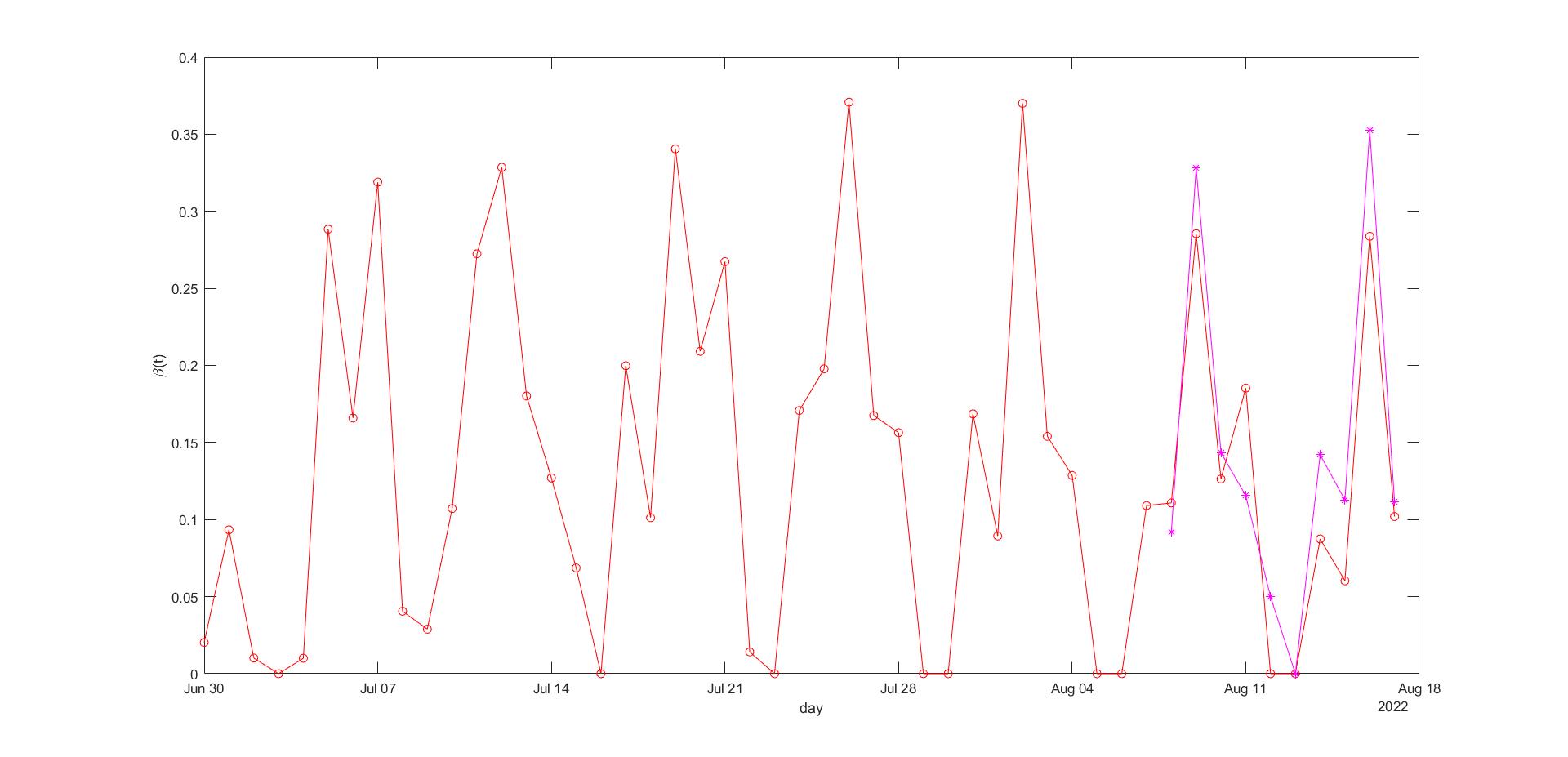}
    \includegraphics[height=3.5cm]{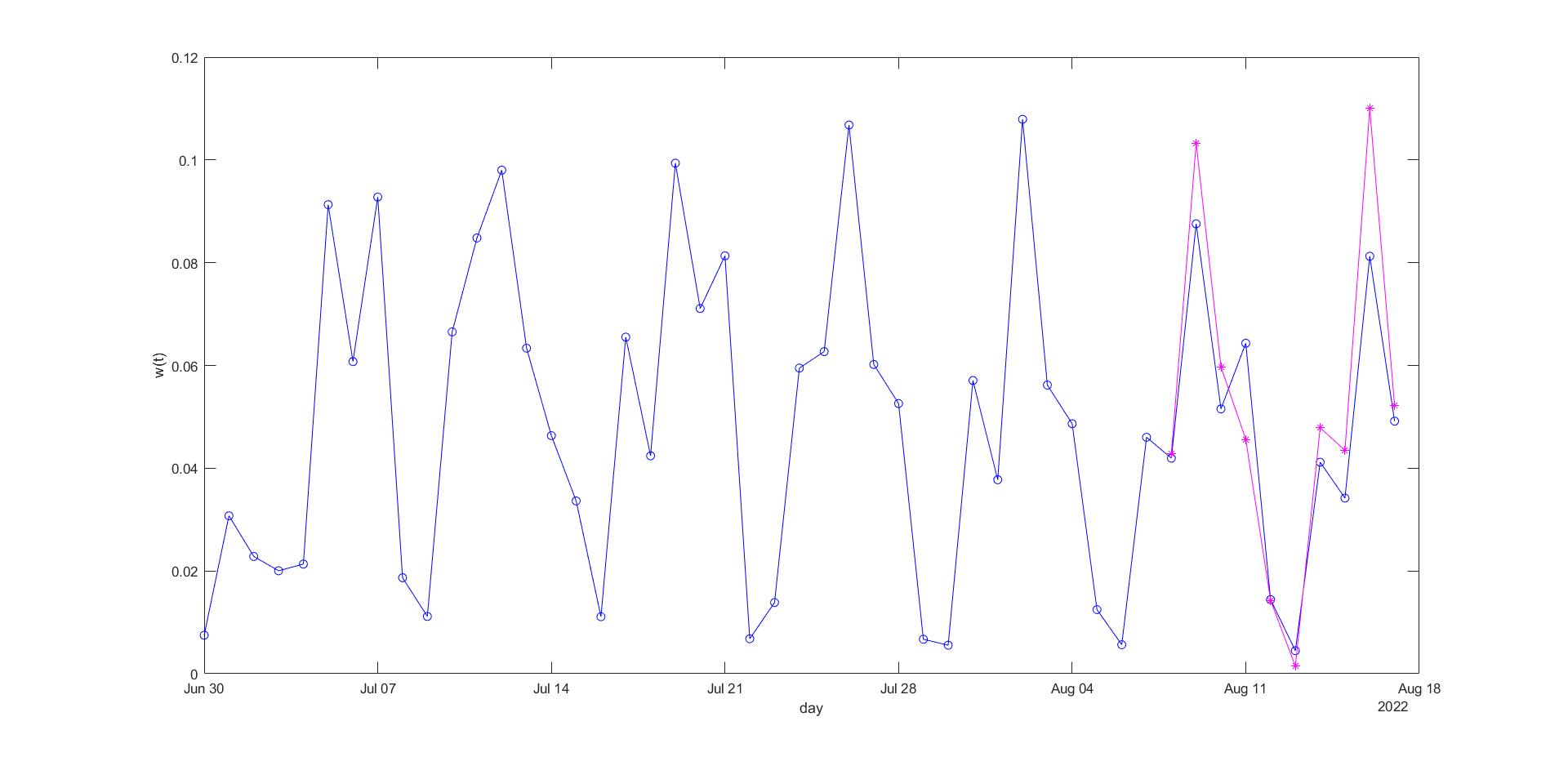}
    \caption{Predictions of $\beta(t)$ (left) and $w(t)$ (right). The hollow points represent the real data and the star-shaped points represent the predicted values.}
    \label{fig:beta}
\end{figure}

\begin{figure}[H]
    \centering
    \includegraphics[height=4cm]{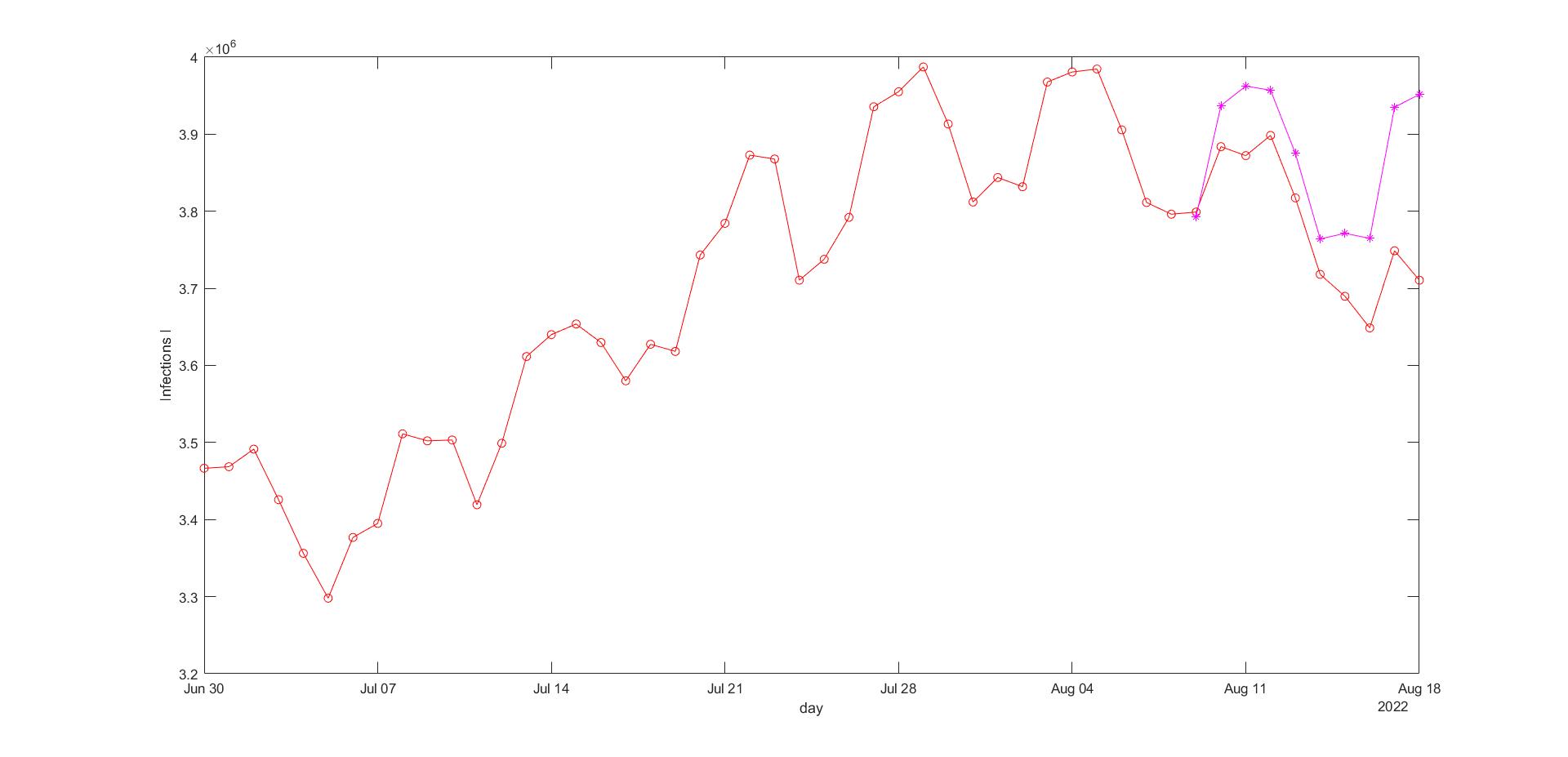}
    \caption{Prediction of $I(t)$. The hollow points represent the real data and the star-shaped points represent the predicted values.}
    \label{fig:I}
\end{figure}

\begin{figure}[H]
    \centering
    \includegraphics[height=4cm]{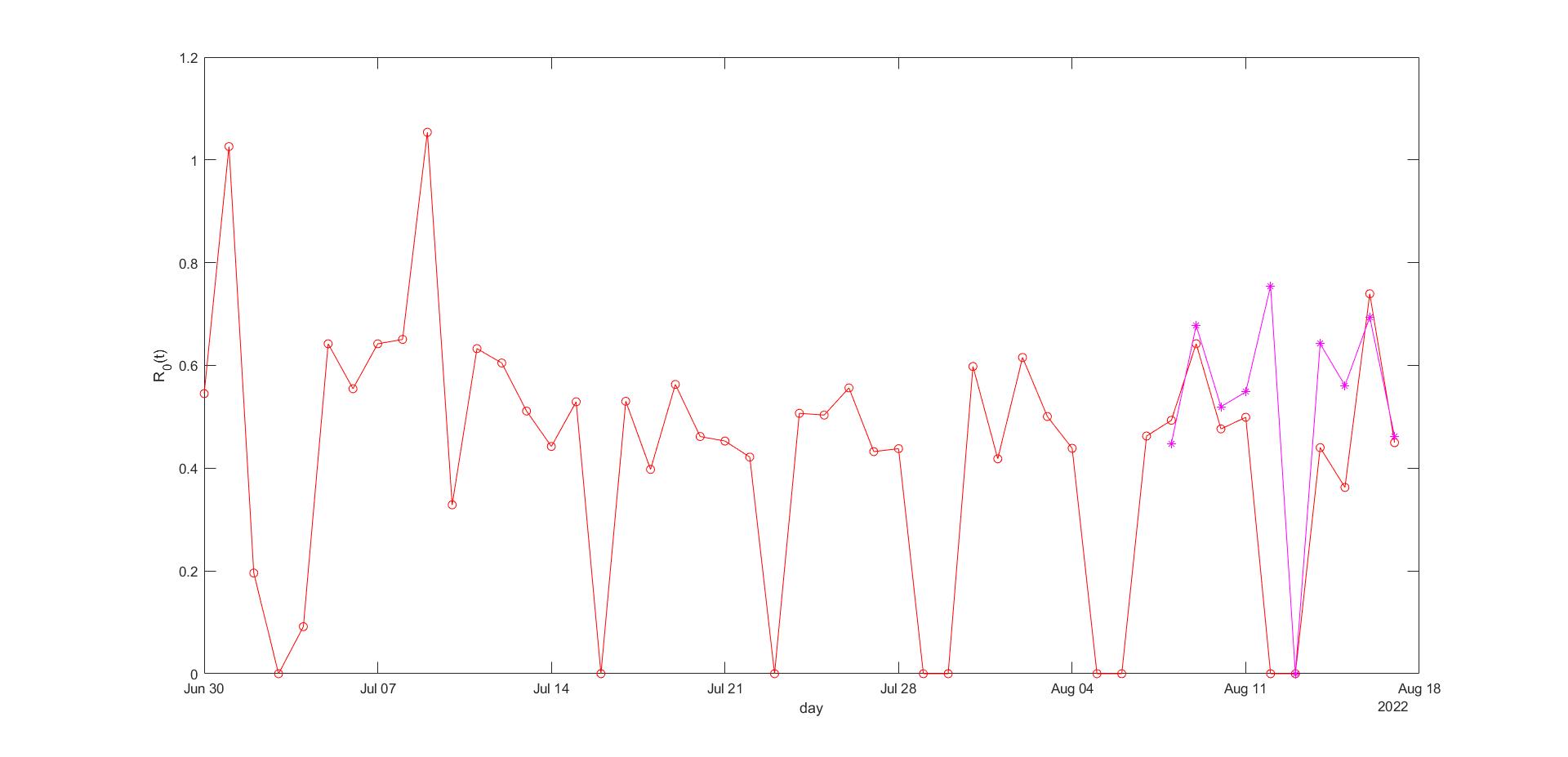}
    \caption{The Time-Varying Basic Reproduction Numbers $R_0(t)$. The hollow points represent the real data and the star-shaped points represent the predicted values.}
    \label{fig:R0}
\end{figure}

\subsection{Distributions With Time-Independent Parameters}\label{secnummartime}\indent\par
From the first forecast in \cref{sectime}, we also get the time-varying rates $\tilde{\lambda}(t),\tilde{\alpha}(t)$ and $\tilde{\theta}(t)$ from August 9-th, 2022 to August 28-th, 2022. Then we estimate $\lambda,\alpha$ and $\theta$ to be the average of the predicted values $\tilde{\lambda}(t),\tilde{\alpha}(t)$ and $\tilde{\theta}(t)$, respectively, namely, $1.3304\times10^{-3},4.5077\times 10^{-4}$ and $1.2460\times10^{-6}$, respectively. Using \cref{eq7}, \cref{eq8} and \cref{eq9}, we compute the probability distributions of the final size and the maximum size and present them in \cref{fig:final size} and \cref{fig:max size}.
\begin{figure}[H]
    \centering
    \includegraphics[height=3.5cm]{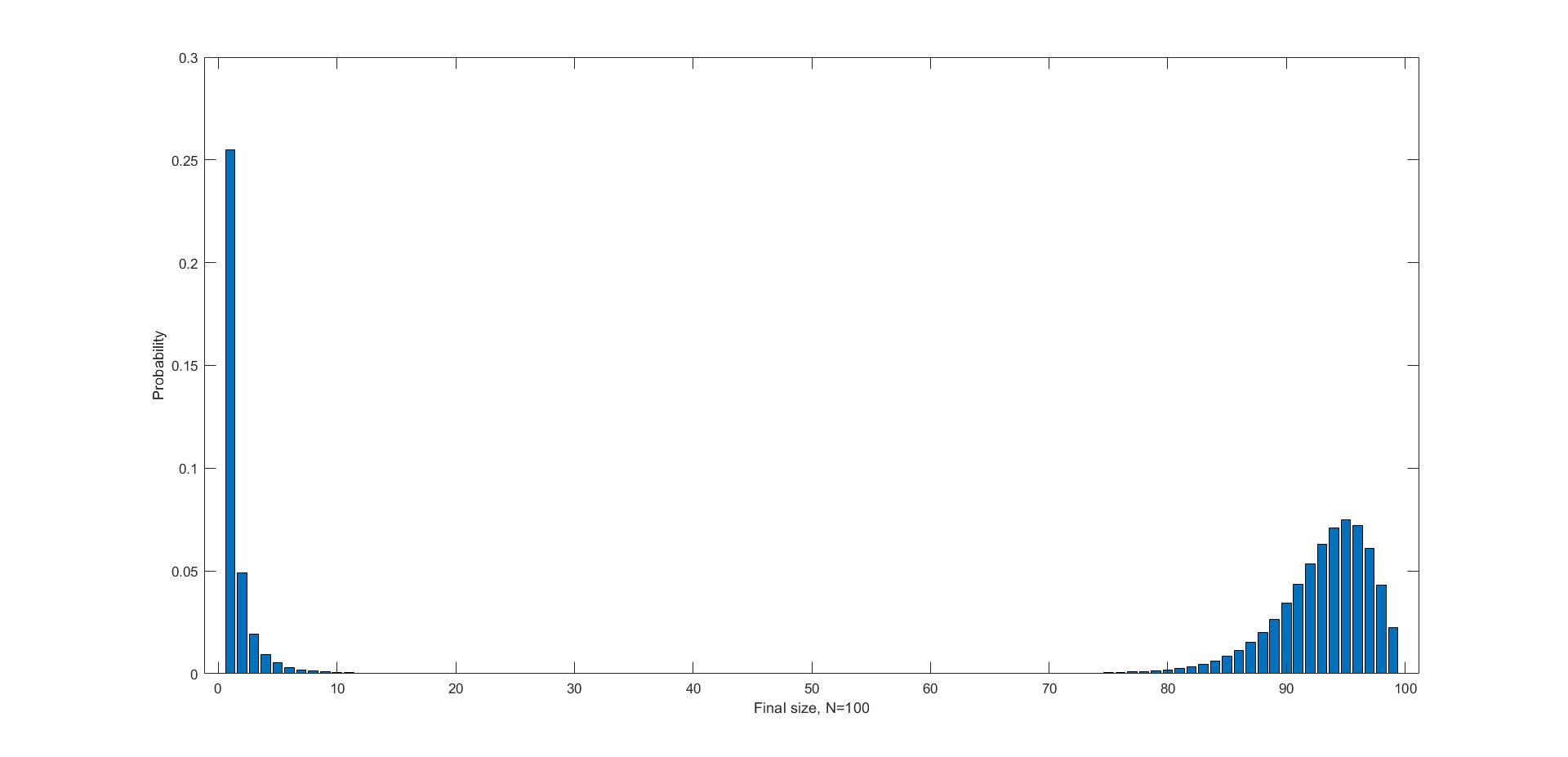}
    \includegraphics[height=3.5cm]{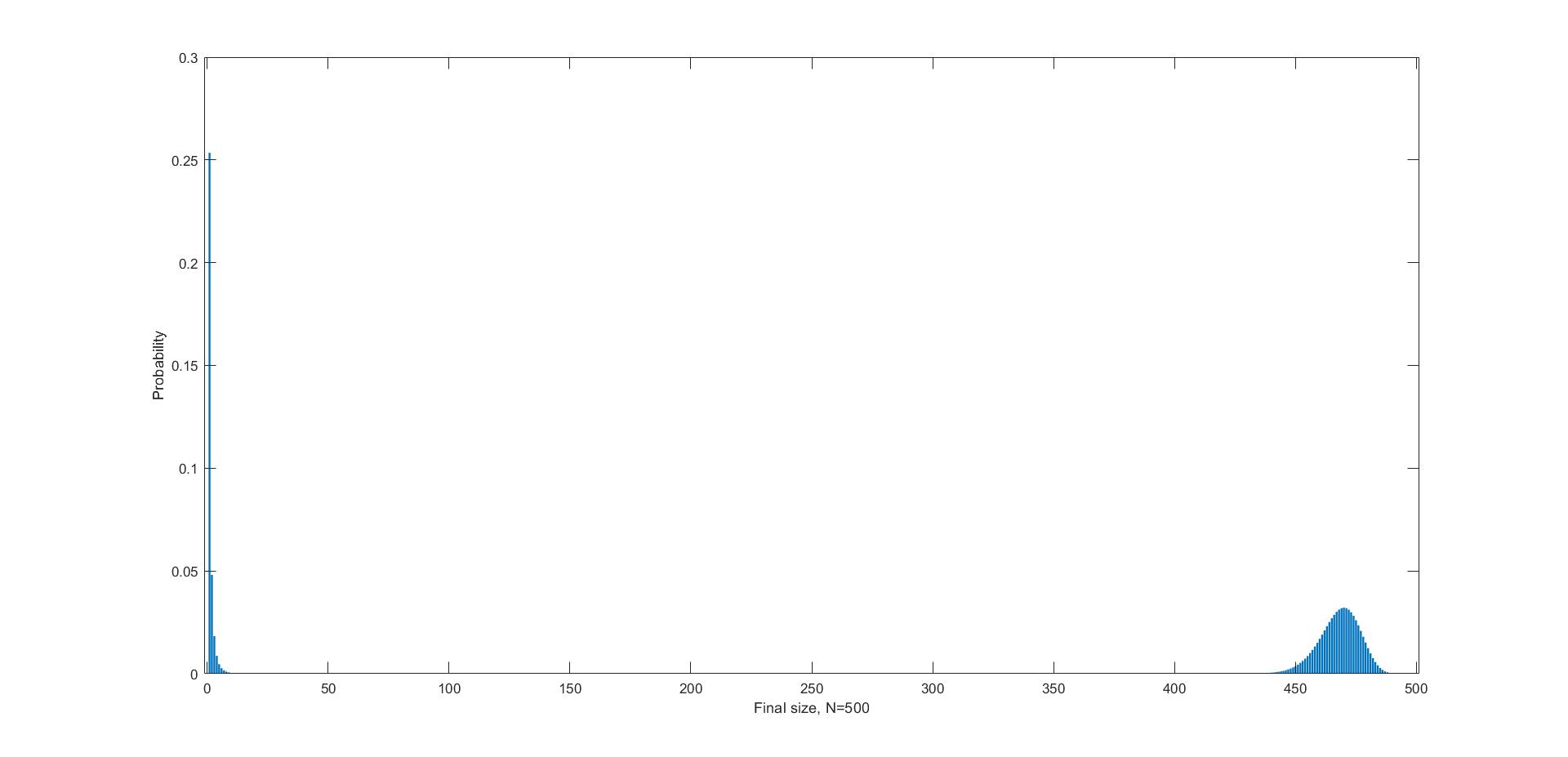}
    \includegraphics[height=3.5cm]{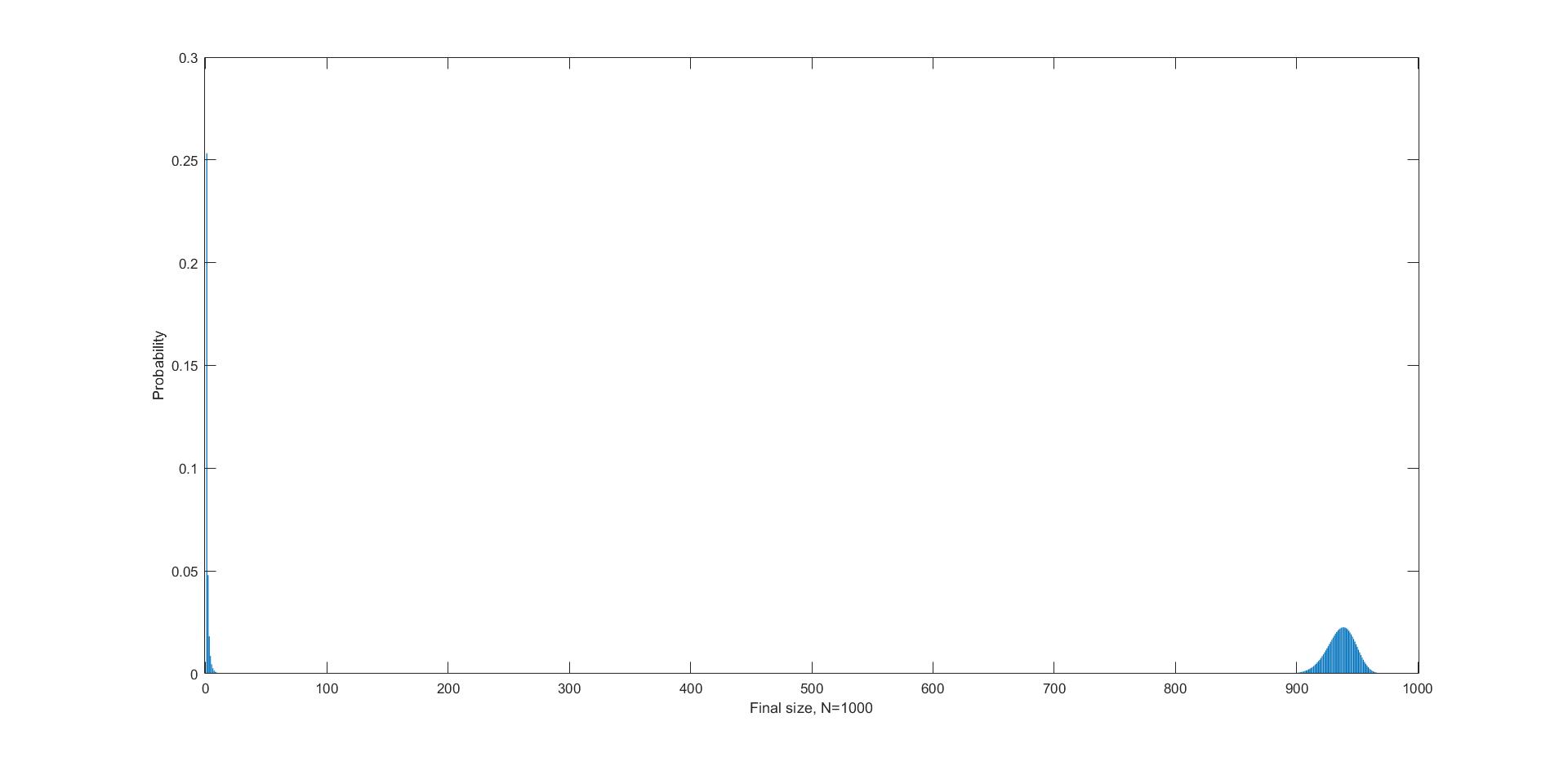}
    \caption{Final Size Distribution with Time-Independent Parameters as $N=100,N=500$ and $N=1000.$}
    \label{fig:final size}
\end{figure}
\begin{figure}[H]
    \centering
    \includegraphics[height=3.5cm]{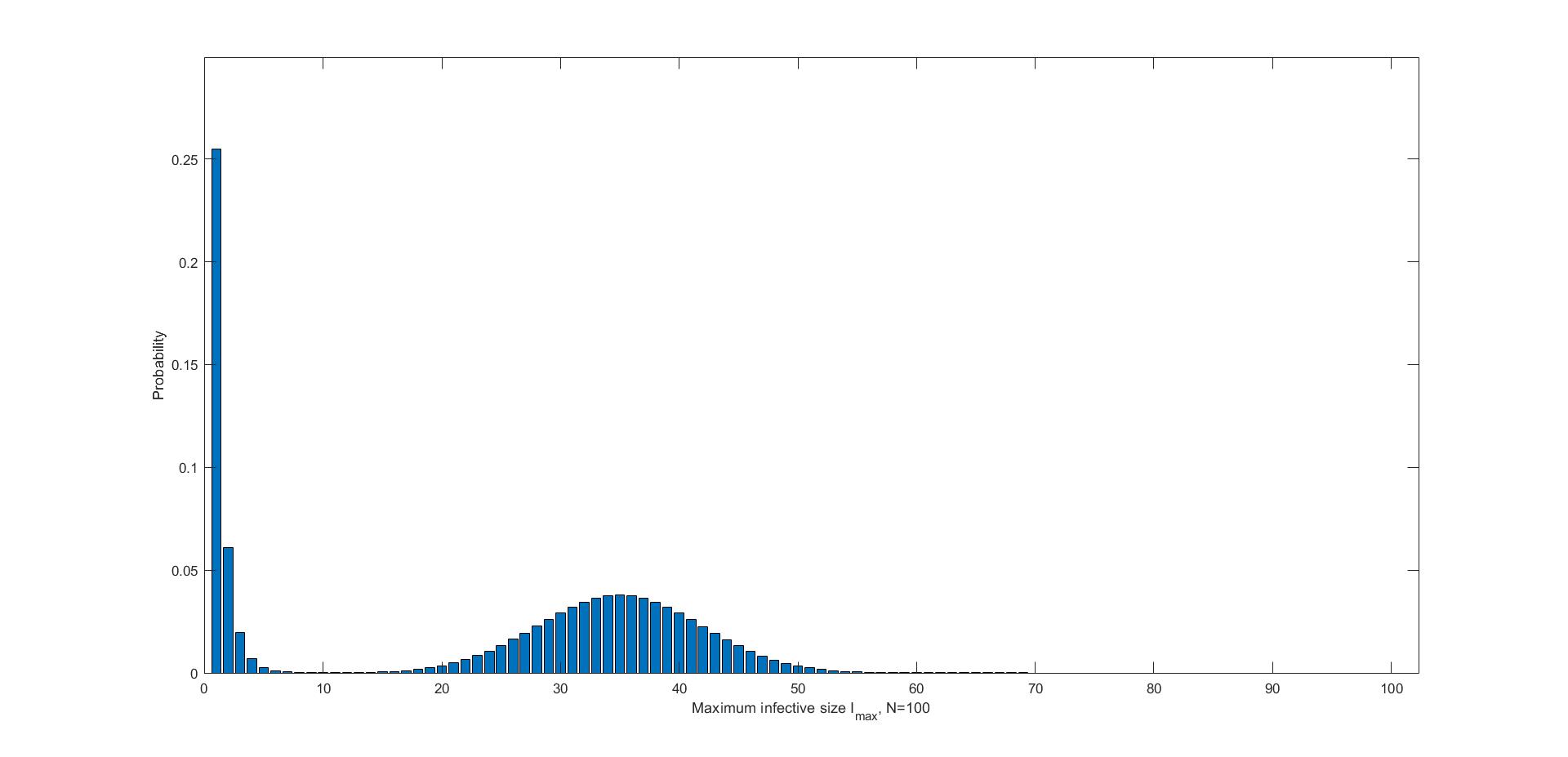}
    \includegraphics[height=3.5cm]{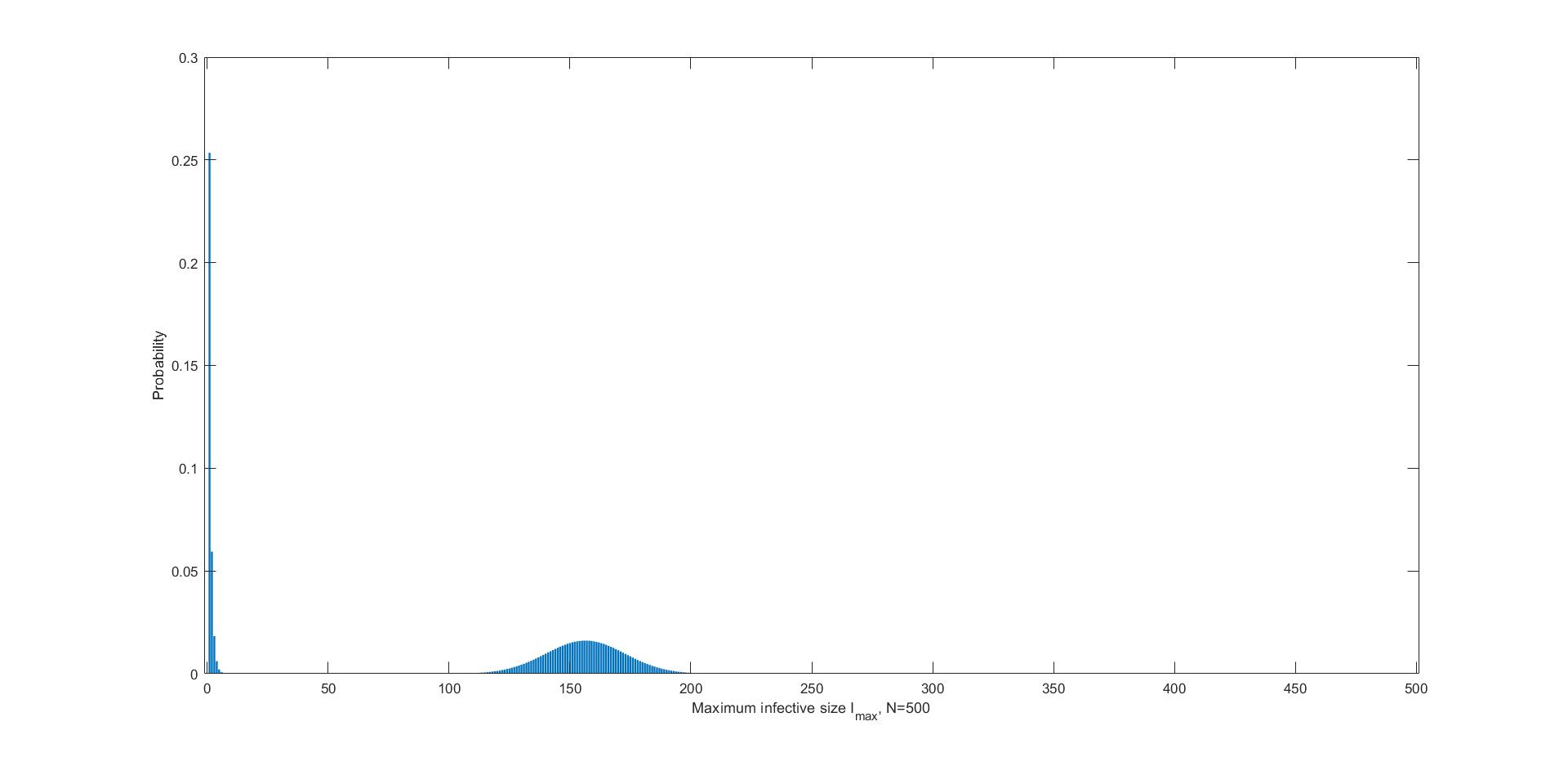}
    \includegraphics[height=3.5cm]{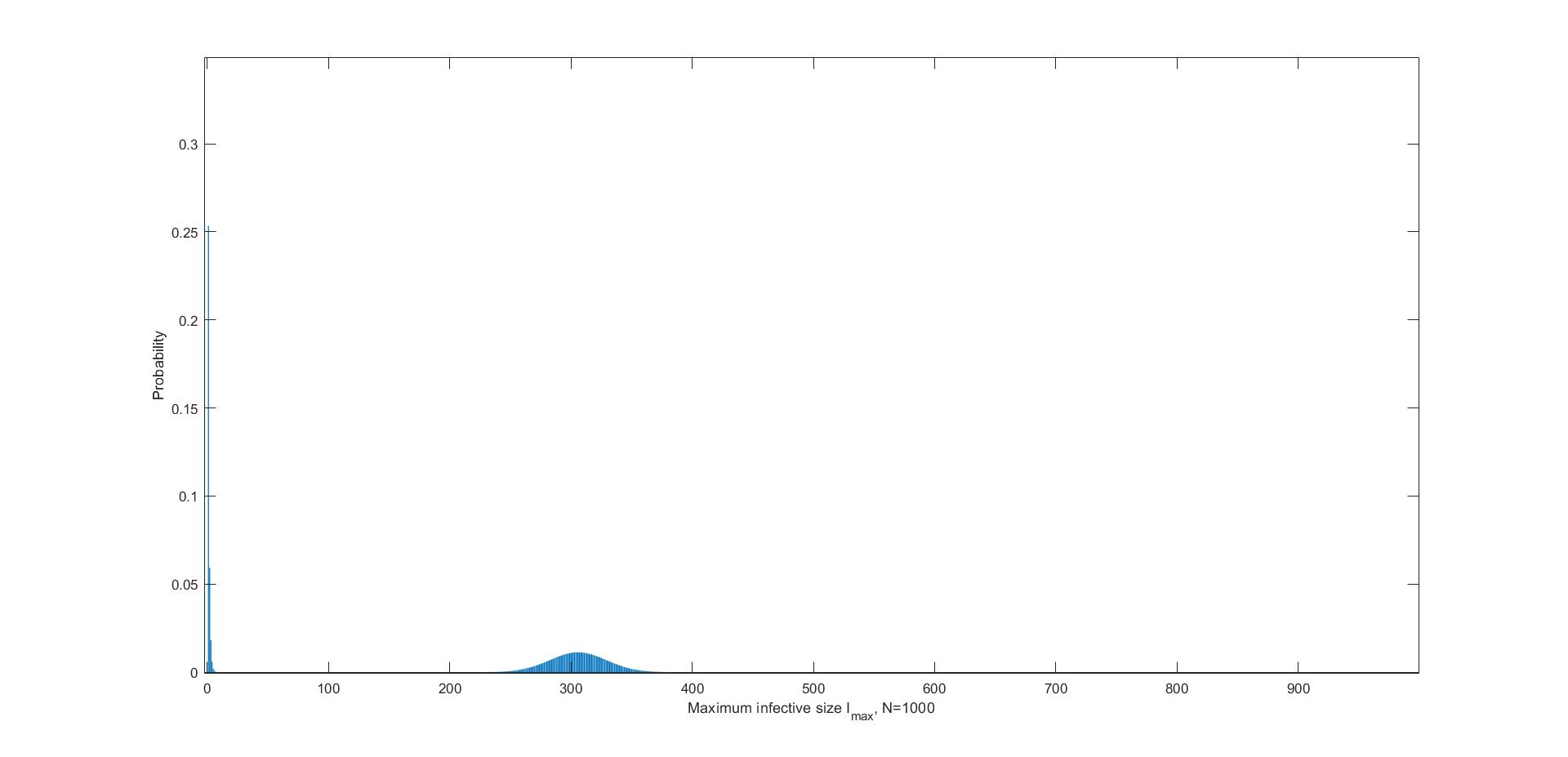}
    \caption{Maximum Size Distribution with Time-Independent Parameters as $N=100,N=500$ and $N=1000.$}
    \label{fig:max size}
\end{figure}

\subsection{Distributions With Time-Varying Parameters}
\indent\par
We next use the time-varying rates $\tilde{\lambda}(t),\tilde{\alpha}(t)$ and $\tilde{\theta}(t)$ from the \cref{next} and apply \cref{algfinal} to compute the probability distributions of the final size and the maximum size under the time-dependent assumption and present them in \cref{fig:timedep final size} and \cref{fig:timedep max size}.
\begin{figure}[H]
    \centering
    \includegraphics[height=3.5cm]{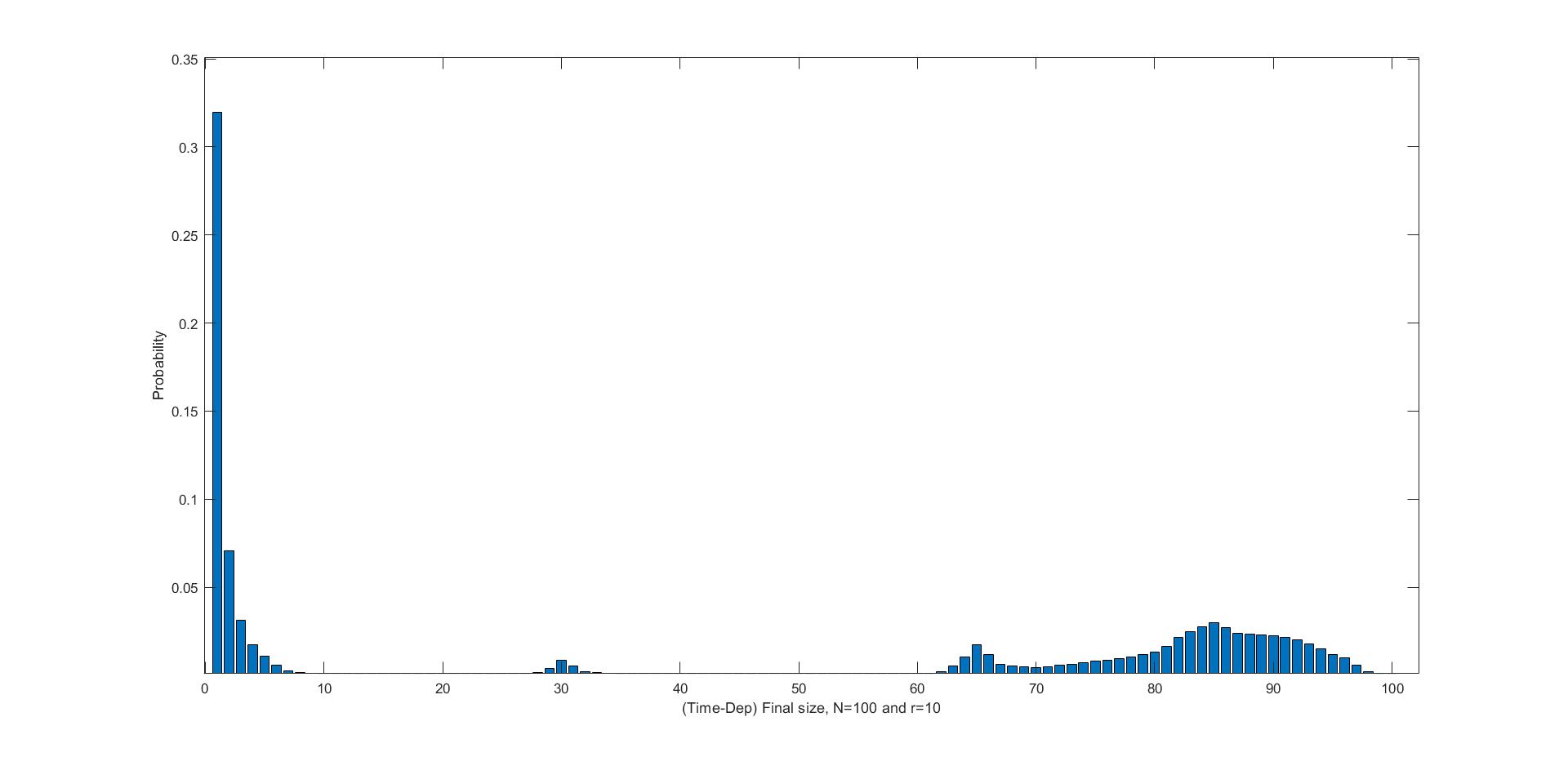}
    \includegraphics[height=3.5cm]{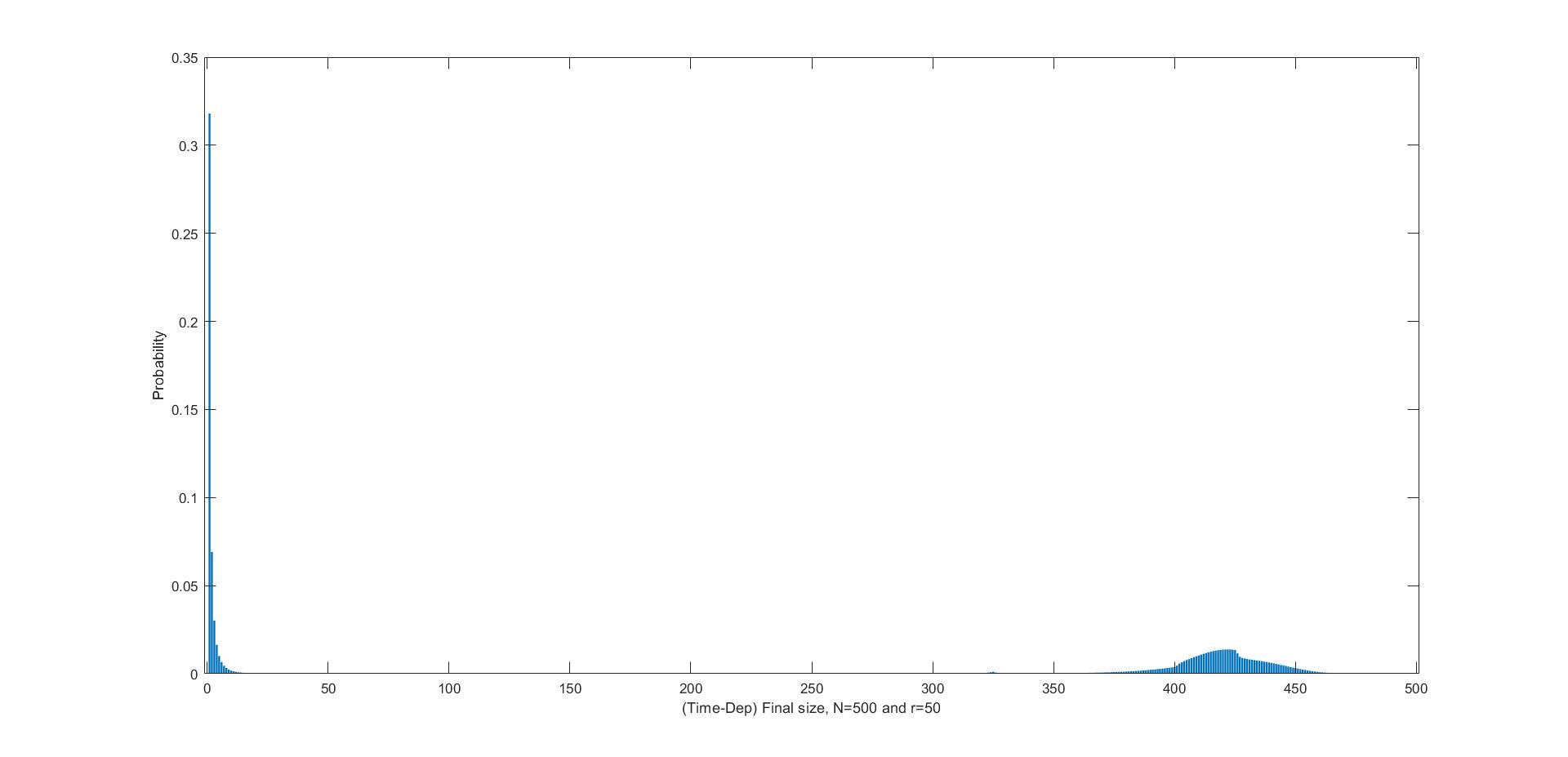}
    \includegraphics[height=3.5cm]{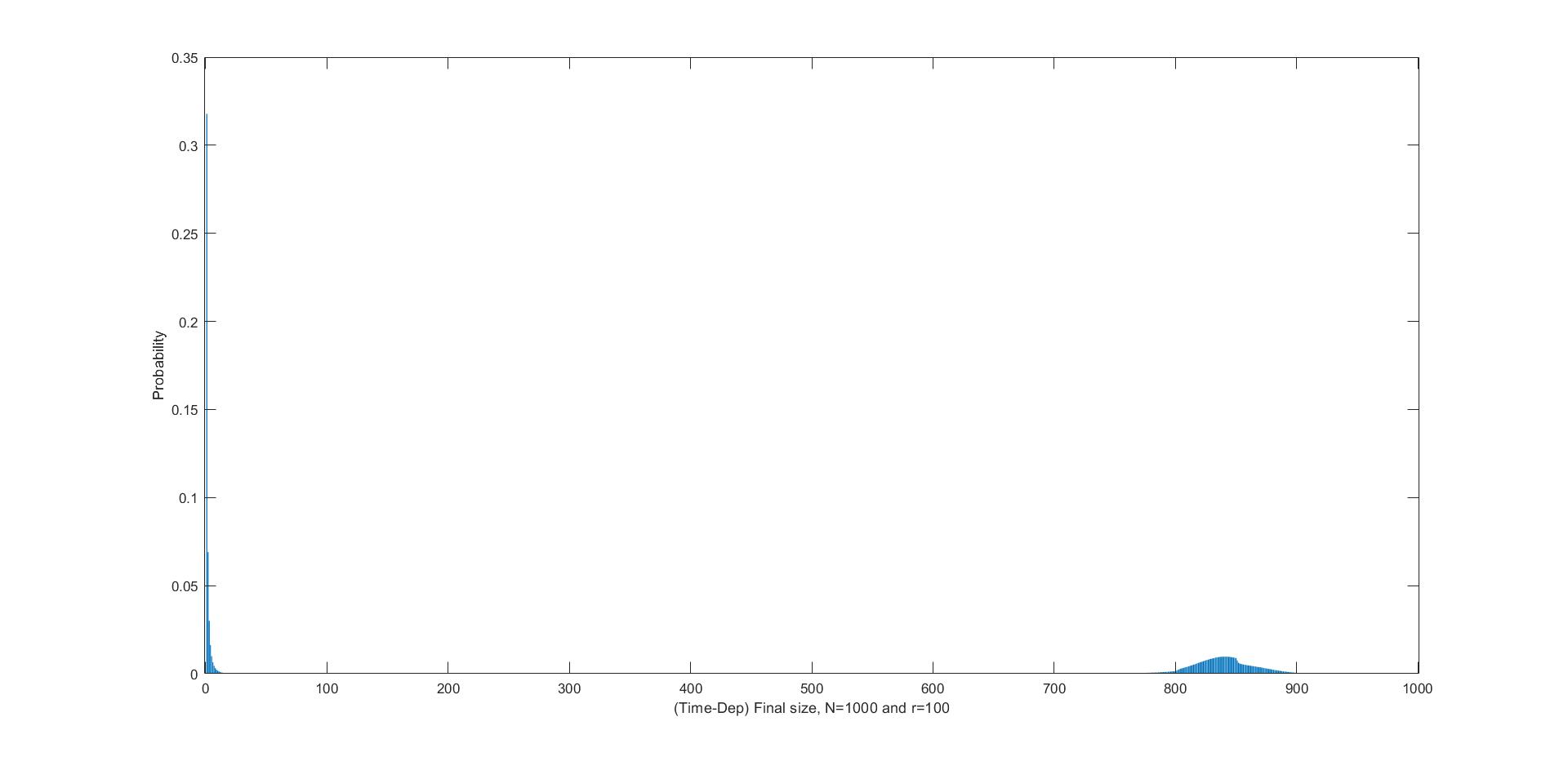}
    \caption{Final Size Distribution with Time-Dependent Parameters as $N=100,N=500$ and $N=1000.$}
    \label{fig:timedep final size}
\end{figure}
\begin{figure}[H]
    \centering
    \includegraphics[height=3.5cm]{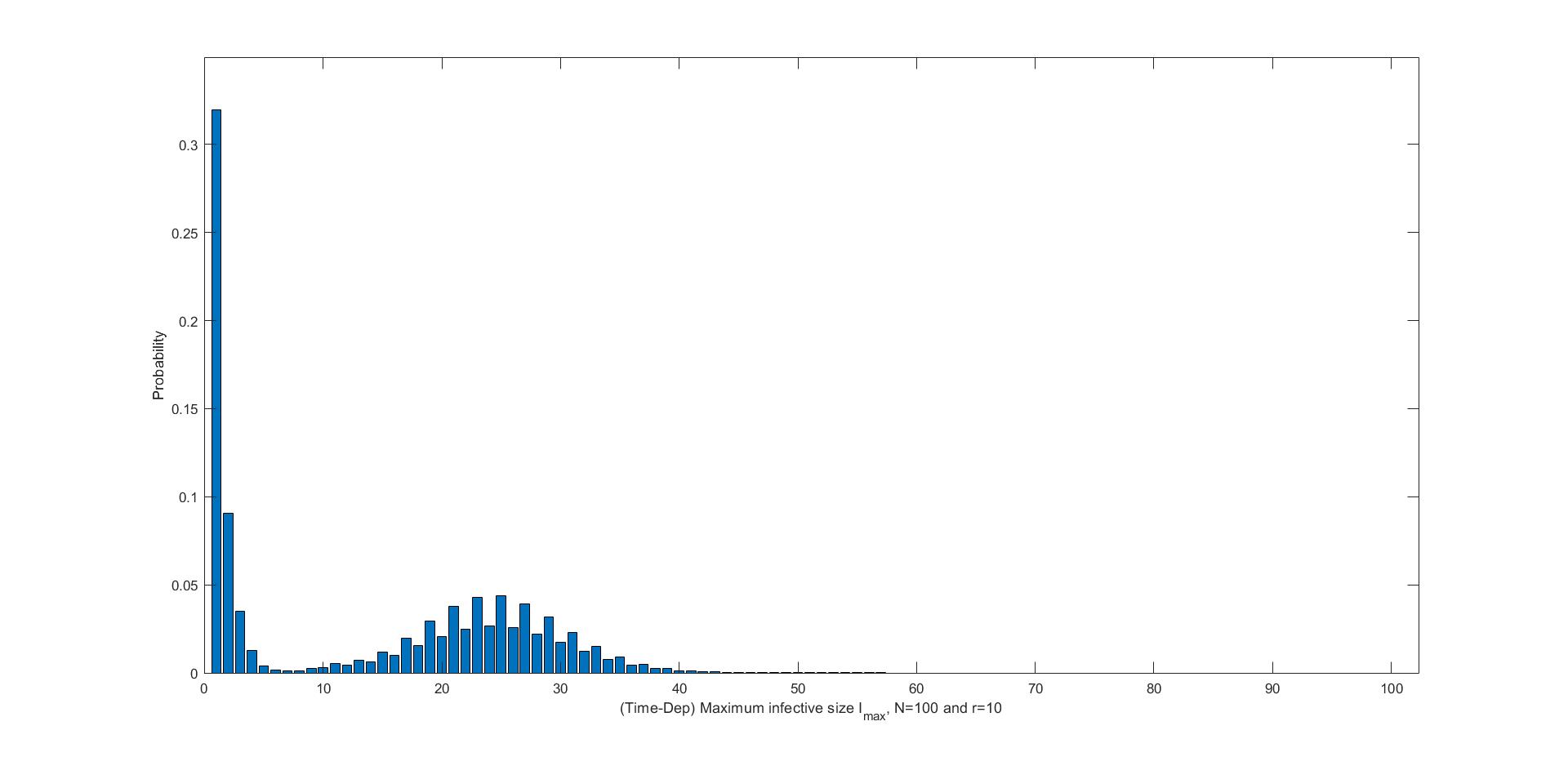}
    \includegraphics[height=3.5cm]{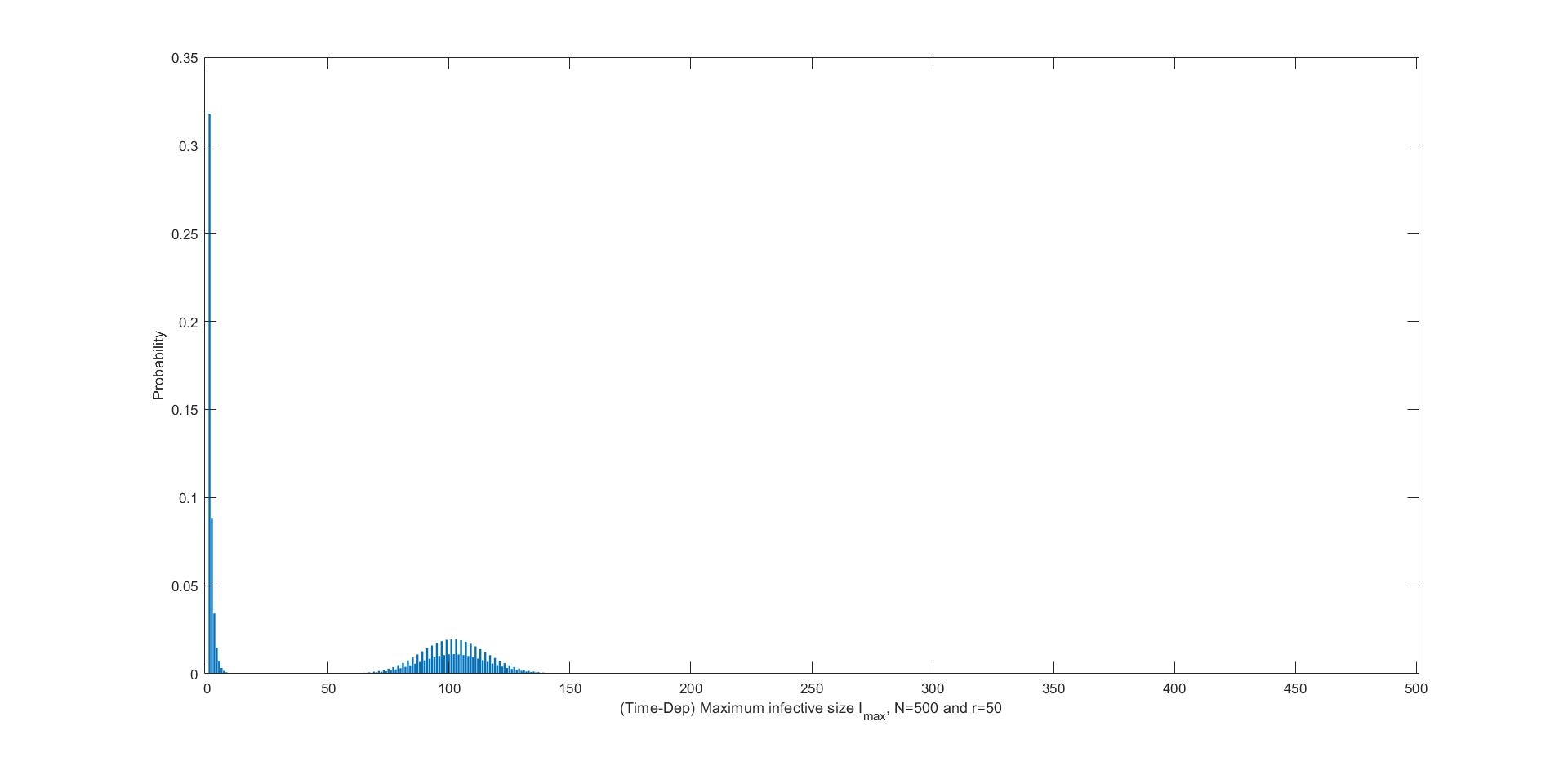}
    \includegraphics[height=3.5cm]{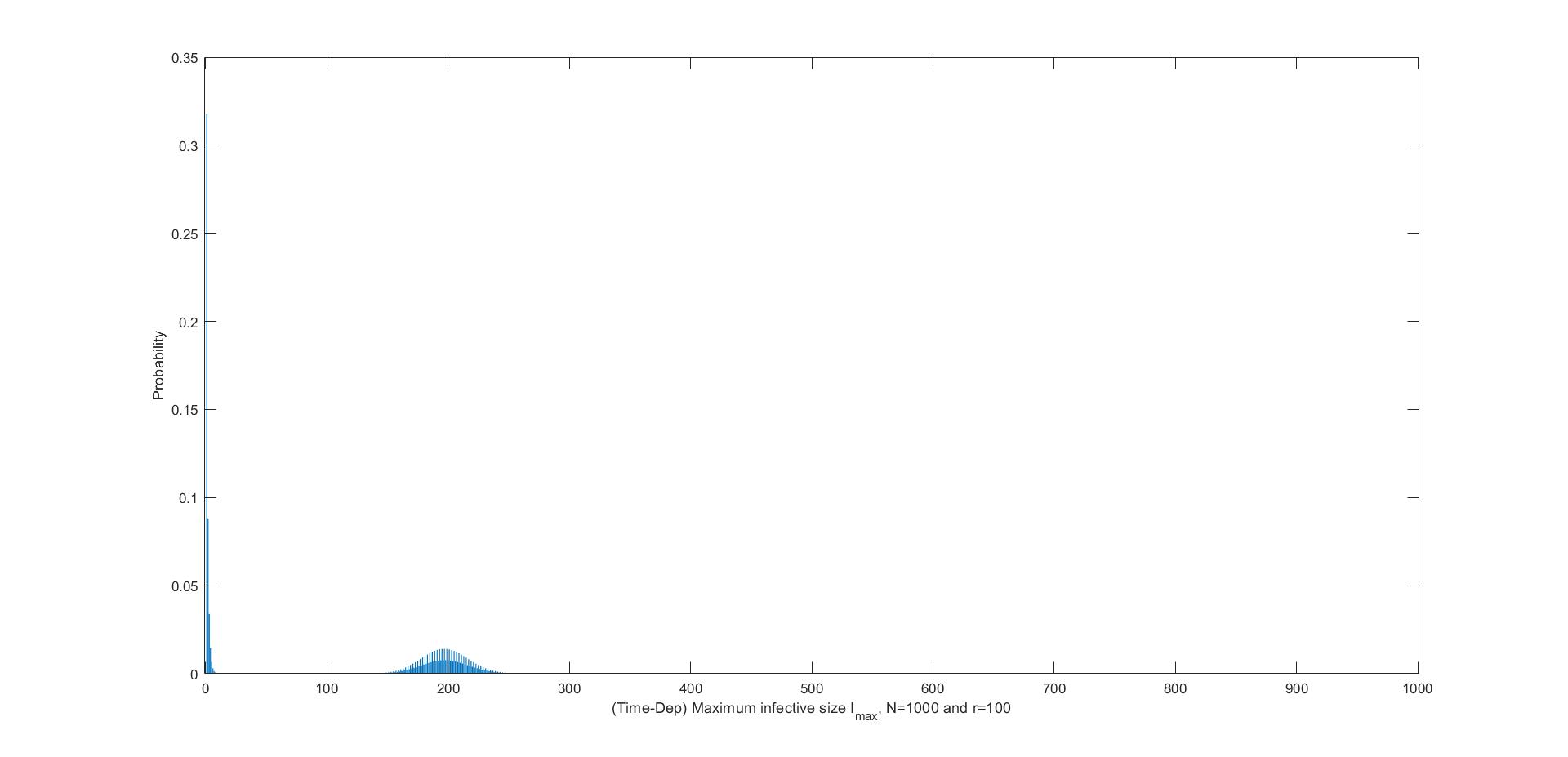}
    \caption{Maximum Size Distribution with Time-Dependent Parameters as $N=100,N=500$ and $N=1000.$}
    \label{fig:timedep max size}
\end{figure}

\subsection{Results and Discussions}\label{dis}
\indent\par The relative prediction errors of $I$ are within $3\%$ in $5$ days and $7\%$ in $10$ days. The trend of change is
closed to the real data of $I$. Therefore, the methods provided in \cite{chen2020time} and \cite{lin2020data} also have good performance in our SAIVRD model and our model meets the real situation and then the predicted parameters are reliable.

We can observe, from \cref{fig:R0}, that the time-varying basic reproduction number in our
SAIVRD model oscillates not far from $1$ and so the epidemic is likely to be lowly prevalent and there may not be another outbreak in the short term.

The distributions of the final size and the maximum size are bimodal, which agrees the results from many previous papers, such as \cite{greenwood2009stochastic}, \cite{black2015computation} and \cite{icslier2020exact}. These mean that with high probability, the epidemic ends and is maximized with very high or very low size. With time-dependent parameters, the distributions look slightly different. There is another peak in the final size distribution, the probability of minor
outbreak is higher and the maximum size distribution is oscillating with time-dependent parameters. As mentioned in \cref{data}, we can observe that the distributions are similar between different populations and so we can estimate the “ratio” of the final size and the maximum size distributions by observing those in small populations. From \cref{fig:final size} and \cref{fig:timedep final size}, under recent transmissibility of this disease in the USA, when an initial infection is introduced into all-susceptible (large) population, ``major outbreak'' \cite{allen2017primer} occurs with around $95\%$ of the population; from \cref{fig:max size} and \cref{fig:timedep max size}, with high probability the epidemic is maximized to around $30\%$ of the population. On the other hand, for the distributions with time-dependent parameters, the probability of ``minor outbreak'' \cite{allen2017primer} is higher than that with time-independent parameters and of course it follows that the probability of major outbreak is lower than that with time-independent parameters.

Finally, we discuss the relation between the final size distribution and the extinction probability with one initial infection in an infinite (or very large) population with exponential distributed infectious period. The cumulative probability distribution of the final size using the settings in \cref{secnummartime} is presented in \cref{fig:cufinalsize}.
\begin{figure}[H]
    \centering
    \includegraphics[height=5cm]{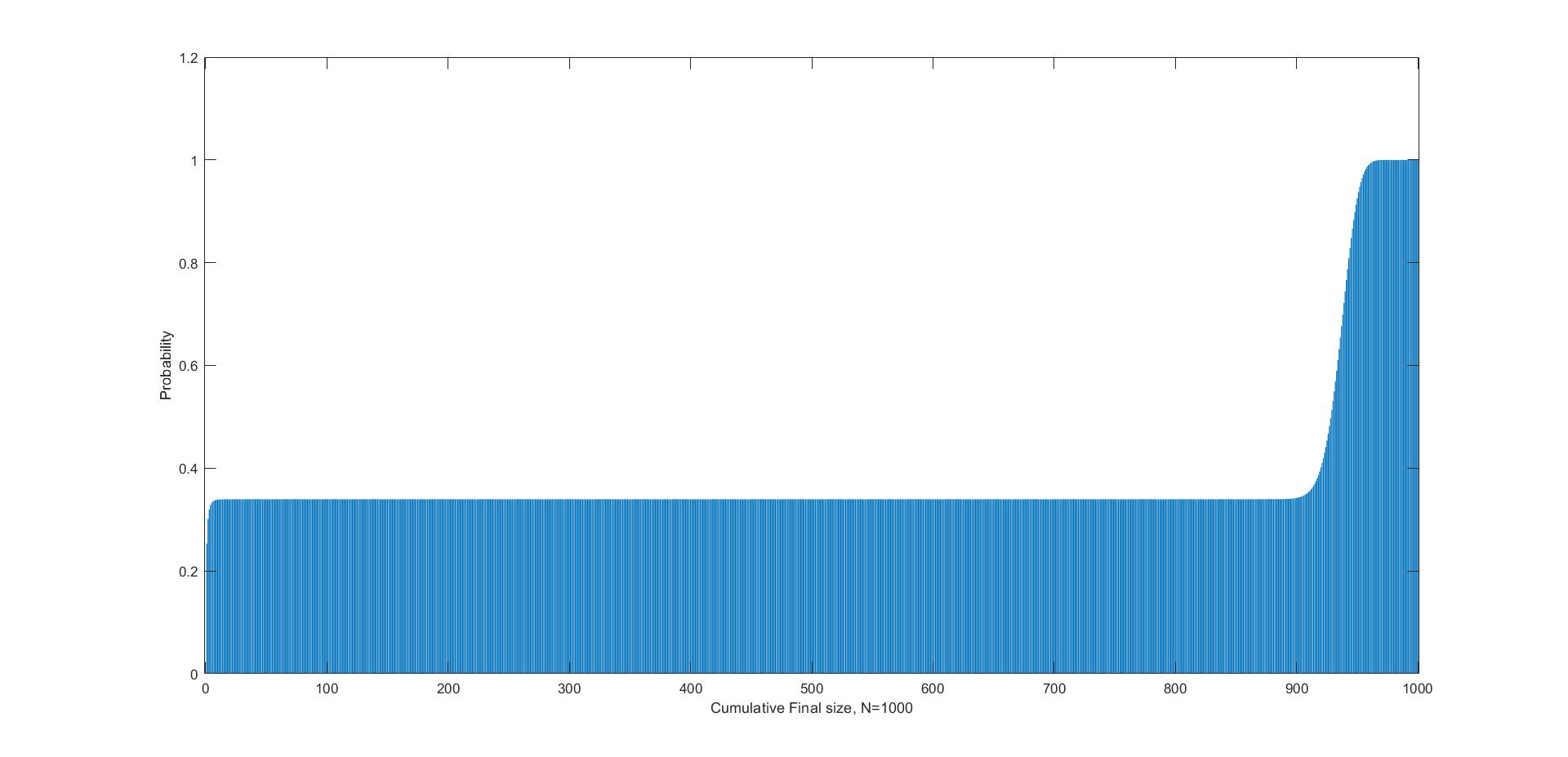}
    \caption{Cumulative Final Size Distribution with Time-Independent Parameters, $N=1000.$}
    \label{fig:cufinalsize}
\end{figure}
 By using the parameters from \cref{secnummartime}, the basic reproduction number in this SARV model is $R_0^{\mathrm{SARV}}=2.9513$ and the extinction probability is $u=0.3388,$ which is approximately the value of the plateau in \cref{fig:cufinalsize} and also $\dfrac{1}{R_0^{\mathrm{SARV}}}$. In view of the derivation in Section 2 in \cite{miller2018primer}, using the  plateau value of the cumulative final
size distribution with $N=1000$ is an ideal simulation for approximation of the extinction probability and the
outbreak probability with one initial infection in a large population.

\section{Conclusion and Future Works}\label{seccon}
\indent\par
At this stage of COVID-19, many countries, such as the UK and the USA \cite{WEBSITE:6} and \cite{WEBSITE:7}, loosen their public health policies against this epidemic, looking ahead to living with COVID-19. This decision may have theoretical basis.  Our discussion of extinction probability is an example. We can see that the extinction probability is as low as about $33.9\%$ by our simulation.

In our paper, we proposed a deterministic SAIVRD model and a stochastic Markov SARV model, simplified from the SAIVRD model, of the epidemic COVID-19. In the deterministic SAIVRD model, we analyzed the existence and the asymptotic stability of disease-free and endemic equilibria related to the basic reproduction number. Based on this model, we conducted the numerical simulations for data forecast to do short-term prediction with time-varying rates and got small relative errors (metioned in \cref{intro} and \cref{dis}), which means that our proposed model meets the real situation in the society of our data set. This forecast can also answer the first question in \cref{intro}. Next, in our stochastic Markov SARV model, we extended the parameters in branching process to be time-dependent and used this model to approximate the probability distributions of the final size and the maximum size. 
These can answer the second question in \cref{intro}. The results with time-dependent and time-independent parameters are a little different and so it is worthwhile to investigate distributions considering the time-dependent parameters instead of time-independent ones. Finally, we estimate the probability of extinction in real situation by both calculating directly and using the cumulative probability of the final size.

In time-dependent Markov SARV model, we assumed the number $r$ of steps in a day. One interesting question comes here: how does the choice of $r$ effect the results?  One possible direction is to estimate the time until the epidemic ends and then we can decide how long do we need to predict in \cref{next} so that $r$ can be chosen so that every step lies within this time. With aid of \cite{parag2020exact} and \cite{elhassan2020mathematical}, estimating the ending time of an epidemic is one of future tasks to make a better choice of $r$ in our time-dependent Markov SARV model.

Furthermore, we would extend our models to many aspects of epidemiology modeling, such as data forecast by stochastic differential equations \cite{atangana2021modeling}, the effect of diffusion \cite{lee2021mean} and even more combining them to conduct models with stochastic partial differential equations to conduct more epidemiologically realistic models.

	\printbibliography
\end{document}